\documentclass[11pt]{article}
\usepackage{amsmath,amssymb,amsfonts,latexsym,multirow}
\usepackage{epsfig,theorem,graphicx,rotating,geometry,color}
\input{epsf.tex}
\usepackage{amsmath,amssymb,amsfonts} 
\usepackage{graphicx,rotating}
\input{epsf.tex}
\usepackage{lineno}

\newenvironment{proof}[1][Proof]{\textbf{#1.} }{\ \rule{0.5em}{0.5em}}
\newtheorem{algo}{Algorithm}
\newtheorem{remark}{Remark}
\newtheorem{theom}{Theorem}

\newtheorem{prop}{Proposition}

\newcommand{\R}{\ensuremath{\mathbb{R}}}
\newcommand{\E}{\ensuremath{\mathbb{E}}}

\newcommand{\Tr}{\ensuremath{\text{tr}}}
\newcommand{\etal}{\emph{et al.}}
\newcommand{\al}{\alpha}

\newcommand{\sig}{\sigma}
\newcommand{\hsig}{\hat{\sigma}}
\newcommand{\hfemp}{\hat{f}_{\text{EMP}}}
\newcommand{\hfkde}{\hat{f}_{\text{KDE}}}
\newcommand{\hflpe}{\tilde{f}_{\text{LOrPE}}}
\newcommand{\hflop}{\hat{f}_{\text{LOrPE}}}
\newcommand{\hfose}{\hat{f}_{\text{OSDE}}}

\newcommand{\hf}{\hat{f}}
\newcommand{\hh}{\hat{h}}
\newcommand{\hr}{\hat{r}}

\newcommand{\xf}{x\sub{fit}}

\newcommand{\Keff}{K\sub{eff}}

\newcommand{\sub}[1]{\ensuremath{_{\mbox{\scriptsize \,#1}}}}

\newcommand{\at}{\tilde{a}\sub{fit}}
\newcommand{\bt}{\tilde{b}\sub{fit}}

\newcommand{\nc}{\normalcolor}

\definecolor{darkgreen}{rgb}{0,0.5,0}

\begin{document}

\title{Local Orthogonal Polynomial Expansion for Density Estimation}
\author{D.P. Amali Dassanayake\thanks{Texas Tech University, Department of Mathematics
\& Statistics, Lubbock, U.S.A.}
  \and Igor Volobouev\thanks{Texas Tech University, Department of Physics, Lubbock, U.S.A.}
\and \and A. Alexandre Trindade$^*$}
\maketitle


\begin{abstract}
A Local Orthogonal Polynomial Expansion (LOrPE) of the
empirical density function is proposed as a~novel
method to estimate the underlying density. The estimate is constructed by matching
localized expectation values of orthogonal polynomials to the values
observed in the sample. LOrPE is
related to several existing methods, and generalizes straightforwardly to multivariate settings. By
manner of construction, it is similar to Local Likelihood Density
Estimation (LLDE). In
the limit of small bandwidths, LOrPE functions as Kernel
Density Estimation (KDE) with high-order (effective) kernels inherently free of boundary
bias, a~natural consequence of kernel reshaping to accommodate endpoints. Faster
 asymptotic convergence rates follow. 
In the limit of large bandwidths,
LOrPE is equivalent to Orthogonal Series Density Estimation (OSDE)
with Legendre polynomials. We compare the
performance of LOrPE to KDE, LLDE, and OSDE, in a~number of simulation
studies. In terms of mean integrated squared error, the results suggest
that with a~proper balance of the two tuning parameters, bandwidth and
degree, LOrPE generally outperforms these competitors when estimating
densities with sharply truncated supports.
\end{abstract}


\maketitle

\section{Introduction}
Few areas of statistical inference receive as much attention as the
classical problem of nonparametric density estimation. Taking as our
basis for inference a~random sample of
observations $x_1,\ldots,x_n$ from an underlying continuous
distribution with probability density function (PDF) $f(\cdot)$
defined on the compact support $[a,b]$, the simplest starting point is the empirical density function (EDF)
\begin{equation}\label{edf-pdf}
\hfemp(x) = \frac{1}{n} \sum_{i=1}^{n} \delta (x - x_i),
\end{equation}
where $\delta(\cdot)$ is the Dirac delta function. If additionally we
assume the existence of a~first few derivatives or that the PDF can have at most a~few modes,  a~convolution
of the EDF with a~kernel function, $K(\cdot)$, often provides a~much
better estimate by producing a~weighted
average of points close to $x$. $K(\cdot)$ itself is usually chosen to be a~symmetric continuous density with a~scale 
parameter, so that the resulting kernel density estimate
(KDE) is
\begin{equation}\label{kde-pdf}
 \hfkde(x)\equiv\int
 \frac{1}{h}K\left(\frac{x-y}{h}\right) \hfemp(y) dy =
 \frac{1}{n}\sum_{i=1}^n\frac{1}{h}K\left(\frac{x-x_i}{h}\right). 
\end{equation} 
The critical KDE tuning parameter is the bandwidth $h$. A convenient,
tractable criterion which is typically 
used to optimize the choice of this parameter is the
mean integrated squared error (MISE).
For large sample sizes, the MISE can be expanded in powers of $n^{-1}$.
The two leading terms in this expansion are associated with
the bias and variance of the estimator. Omission of
all higher-order terms  results in the \emph{asymptotic MISE} (AMISE) approximation.

Under regularity conditions, KDE is consistent, with an AMISE-optimal choice of bandwidth ($h_{\ast}$)
which depends on (computable) kernel moments and the (uncomputable)
integrated squared curvature of $f$. Although the Epanechnikov kernel
minimizes AMISE (is asymptotically optimal),
the choice of kernel is generally not as influential
as the choice of bandwidth. See Silverman (1986), Scott (1992) and Wand \& Jones (1995) for detailed
treatments of the subject, and Sheather (2004), Wasserman (2006,
ch.~6), and  Givens \& Hoeting (2013, ch.~10) for more concise surveys.

Although the optimal $h_{\ast}$ is unattainable in practice, there are several
approaches to dealing with this issue. They range from quick
rules-of-thumb, or \emph{plug-in}
methods,  to the more computationally-intensive bandwidth selection
based on cross-validation (Heidenreich \etal{}, 2013).
Rather, the major drawback of KDE is that
it suffers from boundary
  bias, particularly if $f$ is sharply truncated at the edges of its
  support. In such bounded support settings, KDE fails to attain the optimal
  convergence rate (Jones, 1993). 

One of the earliest attempts at
  correcting this problem was  truncation and reflection of
  boundary kernels (Silverman, 1986). 
Several solutions based on local or adaptive methods have since been
proposed; see for example Malec \& Schienle (2014) for a~survey. A more general solution is to use a~local
polynomial or local likelihood based approach (Hjort \& Jones, 1996,
Loader, 1996, 1999). These methods, and in particular the \emph{local
  likelihood density estimation} (LLDE) detailed in Loader (1999),
  alleviate boundary
  bias, but require the solution of nonlinear equations at each~$x$,
  and  are therefore slow to compute. (Hall \& Tao, 2002, however, argue
that KDE has distinct advantages over LLDE
in the absence of boundary effects.) Although adaptive kernels work
fairly well (e.g., Chen, 1999, Kakizawa, 2004, Jones \& Henderson, 2007), they presume some particular number of derivatives is
matched at the boundary, which affects their asymptotic
performance\footnote{See the R library \texttt{bde} for a~comprehensive implementation of
density estimation methods on bounded supports.}. 



There is thus a~niche to be filled in the nonparametric density
estimation literature by devising methods that  alleviate the
boundary bias issues in a~more general way than the prescribed
corrections of adaptive methods, whilst attaining the optimal KDE convergence
rates in the interior of the support, and yet do all this in a
computationally efficient manner. As will be argued, our
proposed method attains faster asymptotic convergence rates by virtue
of using higher-order (effective) kernels. The initial
motivation for our quest comes
from high energy physics experiments, where there is a~need to
estimate the distribution of visible energy in jets
({\it i.e.}, collections of particles moving in approximately the
same direction) due to smearing by the detector resolution (Volobouev, 2011).
The situation is complicated by the fact that the
energy of any one jet has to be reconstructed from signals produced
by multiple particles in an array of sensors in the measuring
device (calorimeter) with non-linear response (Wigmans, 2000).

It is sometimes possible to use parametric functions to model
such distributions. The results are fair, but there is room for
improvement. Borrowing from the methods in Thas (2010), one idea is to model the bulk of the distribution with
a flexible parametric model (like Johnson curves, Elderton \& Johnson, 1969), and describe the
deviations from this model nonparametrically, in the spirit of Yang
\& Marron (1999). This can be done with
so  called "comparison distributions" (Thas, 2010, ch.~3). The basic
approach is that if $g$ and $G$ denote respectively the PDF and cumulative distribution function (CDF) of
a generic member of the parametric Johnson curves, and if $\psi$ and $\Psi$
denote the PDF and CDF of a~distribution supported on $[0,1]$, then
$F(x)=\Psi(G(x))$ is also a~CDF, with 
\begin{equation} \label{two-step-f}
f(x)=g(x)\psi(G(x)), 
\end{equation} 
as its
corresponding PDF. (This procedure can be iterated given a~sequence of
CDFs $\{\Psi_1,\Psi_2,\ldots\}$  supported on $[0,1]$.)

This suggests one can model observations $\{x_i\}$ from $X\sim f$
by first approximating $f$ with $g$, even if it proves to be
somewhat inadequate, and then mapping the $\{x_i\}$ to the $[0, 1]$
interval according to the transformation, $y_i=G(x_i)$. The density of
the $\{y_i\}$ can now be approximated, either parametrically or
nonparametrically, to yield an estimate of $\psi$, whence the final $f$
is obtained from (\ref{two-step-f}). In the case that $G=F$, the true
CDF, we have of course that $G(X)$ is uniform on $[0, 1]$, a~fact
which can be used to assess the appropriateness of the initial
$G$ ({\it e.g.}, via the comparison
distribution methodology outlined in Thas, 2010, ch.~3). This is
precisely where improved versions of KDE come in; they are needed to
handle the sharply truncated support boundaries of the density of the
$\{y_i\}$ resulting from this approach.

In multivariate problems, an attractive density estimation approach
consists in decomposing the estimated density into the product of the
copula density and of the marginals (Gijbels \& Mielniczuk, 1990). As the
copula density is defined on the unit hypercube, KDE of the copula density
suffers considerably from  boundary bias. While a~number of methods
have been proposed for alleviating this deficiency (as reviewed in
Charpentier \etal{}, 2006; see also Chen \& Huang, 2007),
the asymptotic convergence rate of these methods
at the boundary is nevertheless
inferior to the convergence rate inside the hypercube.

With this backdrop, we propose the use of \emph{local orthogonal polynomial
expansion} (LOrPE) as a~new method to perform nonparametric density
estimation. The theoretical development and genesis of LOrPE is discussed in
section~\ref{sec:lorpe-intro}. Section~\ref{sec:lorpe-others}
discusses connections with other methods: KDE, LLDE, and \emph{orthogonal
      series density estimation} (OSDE). In particular, we
establish there that LOrPE is equivalent to KDE with
a high-order kernel for points well inside the support of the
PDF. Thus, and through appropriate choice of its tuning parameters
(discussed in section~\ref{sec:mod-select}), LOrPE provides a~general way to achieve adaptive (kernel) behavior, while also
attaining optimal asymptotic convergence
rates. Section~\ref{sec:simulate} examines the performance of
LOrPE closely in some simulation studies, in both oracle (best case)
and non-oracle settings, with respect to the competitors
outlined in section~\ref{sec:lorpe-others}. The paper concludes in
section~\ref{sec:real-data} with an illustration on a~real dataset.


\section{Development of LOrPE}\label{sec:lorpe-intro}

LOrPE inherits several of its features from OSDE (Efromovich, 1999),
      and can in fact be thought of as a
      localized version of OSDE. With $\hf(x)$  a~simple initial estimator
      such as
      (\ref{edf-pdf}), LOrPE amounts to constructing a~truncated
      orthogonal polynomial series expansion for the EDF near each point $x\sub{fit}$
where the density estimate is desired. (In practice, these points would usually
be taken to be uniformly spaced on a~grid of values covering the
support of the density.) For a~chosen bandwidth
$h$, this expansion is
\begin{equation}
\label{eq:expansion}
\hflpe(x) = \sum_{k=0}^{M}c_{k}(x\sub{fit}, h) P_{k}\left(\frac{x - \xf}{h}\right),
\end{equation}
where the polynomials $P_{k}(x)$ are constrained to satisfy the normalization condition
\begin{equation}\label{eq:norm0}
\frac{1}{h} \int_{a}^{b} P_{j}\left(\frac{x - \xf}{h}\right)P_{k}\left(\frac{x - \xf}{h}\right) K\left(\frac{x - \xf}{h}\right)dx = \delta_{jk},
\end{equation}
which, with  $\at=(a - x\sub{fit})/h$ and  $\bt=(b - x\sub{fit})/h$, is equivalent to
\begin{equation}
\label{eq:norm}
\int_{\at}^{\bt} P_{j}(y)P_{k}(y)K(y)dy
= \delta_{jk},
\end{equation}
where $\delta_{jk}$ is the Kronecker delta
and $K(\cdot)$ a~suitably chosen kernel function. The coefficients $c_{k}(x\sub{fit}, h)$
are determined by
\begin{equation}\label{eq:convol}
c_{k}(x\sub{fit}, h) = \frac{1}{h} \int \hf(x)P_{k}((x - x\sub{fit})/h)K((x - x\sub{fit})/h)dx,
\end{equation}
which, for $\hf(x)=\hfemp(x)$, is equivalent to
\begin{equation}\label{eq:emp-ck}
c_{k}(x\sub{fit}, h) = \frac{1}{n h}\sum_{i=1}^{n}P_{k}((x_{i} - x\sub{fit})/h) K((x_{i} - x\sub{fit})/h).
\end{equation}
Because negative values can occur, the proposed density estimate at
$x=\xf$ is then $\max\{0, \hflpe(x\sub{fit})\}$. In general, this does
not result in a~\emph{bona fide} density function (similarly to
OSDE), and thus the final step in the process involves performing a
renormalization over all grid points. The final (genuine) density
estimate at $x$ is denoted by $\hflop(x)$. Generalizing LOrPE to
a multivariate setting is in principle straightforward, necessitating
only a~switch to multivariate orthogonal polynomial systems. 

Equation (\ref{eq:expansion}) can be usefully generalized to include a~{\it taper function} $t(k)$ as follows:
\begin{equation}\label{eq:lorpe}
\hflpe(x) = \sum_{k=0}^{\infty} t(k) c_{k}(x\sub{fit}, h) P_{k}((x - x\sub{fit})/h).
\end{equation}
The idea of the taper function is to suppress
high order terms gradually, instead of using a
sharp cutoff at $M$. Also, as will be discussed in section~\ref{sec:mod-select}, a~particular definition
of the taper function allows for a~simple extension of (\ref{eq:expansion}) to
non-integer values of $M$. We will normally require that
$t(0) = 1$ in order to ensure correct asymptotic
normalization, in addition to specifying that $t(k) = 0$ for $k >
M$. 

LOrPE admits an appealing interpretation in terms of the local density
expansion (\ref{eq:expansion}), in which the ``localized'' expectation values of the
orthogonal polynomials $P_k(\cdot)$ are matched to their empirical values calculated
from the data sample. This heuristic interpretation can be understood by making the
following observation. Define the
\emph{localized expectation} (at $x\sub{fit}$) of
a function $\phi$ with respect to kernel $K$ (bandwidth $h$) for a~random variable $X\sim
f$ as,
\[ \E^{(\text{loc})}_{f}[\phi(X)]=\frac{\int \phi(x)K(x)f(x)dx}{\int K(x)f(x)dx}. \]
Then, upon setting  $\phi(x)=P_k(x)$, note that
\[ \E^{(\text{loc})}_{\hfemp}[P_k(X)]=\frac{c_{k}(x\sub{fit},
  h)}{c_{0}(x\sub{fit}, h)} = \E^{(\text{loc})}_{\hflpe}[P_k(X)]. \]

\section{Connections With Other Methods}\label{sec:lorpe-others}
This section explores the connections between LOrPE and KDE, OSDE, and
LLDE. We will show that under certain conditions LOrPE is
essentially equivalent to KDE (Theorem~\ref{th:lorpe-facts}); while under other conditions its
behavior mimics OSDE (Theorem~\ref{th:osde}). Also, the local adjustments instituted
by LOrPE to reduce support boundary bias are very much in the spirit of LLDE.

\subsection{Kernel density estimation}
In
general, LOrPE behaves as a~linear combination
of KDEs with varying kernels. 
To see  this, define $K_k(z)= P_k(z)K(z)$, and note that from (\ref{eq:convol}) with $\hf(x)=\hfemp(x)$ we can
write the expansion coefficients as
\begin{eqnarray*}
c_{k}(x\sub{fit}, h) =\int \frac{1}{h} K_{k}\left(\frac{x -
x\sub{fit}}{h}\right) \hfemp(x)dx \equiv  \hfkde(x | h,K_k),
\end{eqnarray*}
where the notation $\hfkde(x|h,K)$ emphasizes the dependence on bandwidth $h$ and
kernel $K$. Thus (\ref{eq:expansion}) can be written as a~weighted linear combination
of KDEs with varying (improper) kernels $K_k$,
\begin{equation}\label{eq:lorpe-kde-combo}
\hflpe(x) = \sum_{k=0}^{M} \hfkde(x | h,K_k)P_{k}\left(\frac{x -
    \xf}{h}\right). 
\end{equation}

The following proposition establishes a~basic result concerning the
families of orthogonal polynomials arising from commonly used kernels.
\begin{prop}\label{rk:gegen}
For commonly used kernels from
the Beta family supported on $[-1,1]$ (Epanechnikov, Biweight, Triweight, {\it etc}.), condition (\ref{eq:norm})
generates the normalized Gegenbauer polynomials (up to a~common multiplicative
constant) at grid points $\xf$ sufficiently deep inside the support
interval, provided $h$ is small enough to guarantee that
$\at\leq -1$ and $\bt\geq 1$. 
\end{prop}
\begin{proof}
By definition, the normalized Gegenbauer polynomials,
$P_j^{(\alpha)}(x)$, $j=0,1,\ldots$, are orthogonal on $[-1,1]$ with
respect to the weight function $w(x)=(1-x^2)^{\al-1/2}$, for some
$\al\geq -1/2$. This means that 
\begin{equation}\label{eq:gegen}
\int_{-1}^1
P_j^{(\alpha)}(x)P_k^{(\alpha)}(x)w(x)\,dx = \delta_{jk}. 
\end{equation}
Noting that $w(x)=c_\al K(x)$, where $K(x)=c_\al^{-1}(1-x^2)^{\al-1/2}I_{[-1,1]}(x)$ is a
beta kernel with associated normalizing constant
$c_\al = \Gamma(\al+1)/[\sqrt{\pi}\Gamma(\al+1/2)]$, equation (\ref{eq:gegen}) becomes
\begin{eqnarray}
\delta_{jk} &=& \int_{-1}^1  P_j^{(\alpha)}(x)P_k^{(\alpha)}(x)c_\al K(x)\,dx
= c_\al \int_{\at}^{\bt}
P_j^{(\alpha)}(x)P_k^{(\alpha)}(x)K(x)\,dx,\nonumber
\end{eqnarray} 
since $K(x)=0$ outside of $[-1,1]$ and $\at\leq -1$ and
$\bt\geq 1$. This requires extending the polynomials so that they are also defined on $[\at,\bt]$.
While this extension is not unique, any reasonable definition
will do, \emph{e.g.}, by using the same coefficients
as on the $[-1, 1]$ interval. Values of $\al=3/2,5/2,7/2,9/2$ define
respectively the Epanechnikov, Biweight, Triweight, and Quadweight kernels. 
\end{proof}

\begin{remark}
If $\xf$ is
sufficiently close to the ends of the support $[a,b]$ relative to the
kernel support, then, since the
kernel is used as the weight function in generating them, the
polynomials will vary depending on $\xf$, and the notation
$P_{k}(\cdot,\xf)$ would be more appropriate. This in turn implies the
kernels $K_k$ in (\ref{eq:lorpe-kde-combo}) also depend on  $\xf$, and
will undergo adjustments near the boundary. For example,
with the Beta kernels of Proposition \ref{rk:gegen}, the effective
support of $K_k$ becomes $[\max(-1,\at), \min(1,\bt)]$.
\end{remark}

The following theorem establishes the main
result that, when evaluated
at grid points far
from the support boundaries,  LOrPE is equivalent to KDE with
a high-order kernel. In particular, this results implies that (under
the appropriate conditions) LOrPE
enjoys the same asymptotic optimality results as does KDE. Unlike KDE
however, LOrPE does not intrinsically suffer from boundary bias because the
orthogonality requirement imposed by
(\ref{eq:norm}) automatically adjusts the shape of 
the (orthogonal) polynomials near the boundary.

\begin{theom}\label{th:lorpe-facts}
When evaluated at points $x\sub{fit}$, (\ref{eq:lorpe})
is equivalent to KDE with the effective
kernel
\begin{equation}\label{eq:effk}
K\sub{eff}(x) = \sum_{k=0}^{\infty} t(k) P_{k}(0) P_{k}(-x) K(-x).
\end{equation}
Under the following additional Assumptions:
\begin{itemize}
\item[(a)] $K(x)$  is an even kernel supported
on some interval $(-a_K,a_K)$ that is symmetric about 0;
\item[(b)] $x\sub{fit}$ is sufficiently far from the density support
boundaries $[a,b]$ so that the $P_{k}(\cdot)$'s can be generated on an interval
of orthogonality that is symmetric about zero,  and subsequently
extended to $[\at,\bt]$ by keeping the same coefficients, where $\at\equiv(a-\xf)/h$ and
      $\bt\equiv(b-\xf)/h$, as in the proof of Proposition~\ref{rk:gegen}; and
\item[(c)] we have $\at\leq -a_K<a_K\leq\bt$;  
\end{itemize}
then the  effective kernel (\ref{eq:effk}):
\begin{itemize}
\item[(i)]  is an even function supported on $(-a_K,a_K)$;
\item[(ii)]  is normalized provided  $t(0) = 1$; and
\item[(iii)] is a~high-order kernel if $t(k)$ is a~step function,
  i.e.~$t(k) = 1$ for all $k \le M$ and $t(k) = 0$ for all $k > M$, in
  which case the kernel order is $M+1$ if $M$ is odd and $M+2$ if $M$ is
      even.
\end{itemize}
\end{theom}

\begin{proof}
See  the appendix.
\end{proof}

The local adjustments made by LOrPE near the support boundary are
illustrated in Figure~\ref{fig:eff-kernel-plots}. In these plots, the
effective kernel $K\sub{eff}$ is shown vs.~$x-\xf$ for a~density that is
sharply truncated at $0$. The normal
density is used as the weight function, with bandwidth set at
$h=0.1$. Polynomials up to degree $M=4$ are considered. The plots
correspond to LOrPE density estimation on the $[0,1]$ interval for
points: exactly at the boundary (left panel), close to the boundary
(middle panel), and away from the boundary (right panel).
\begin{figure}[t]
\begin{center}
\begin{tabular}{ccc}
(a) & (b) & (c) \\
\vspace{-1cm} \\
\includegraphics[scale=0.36]{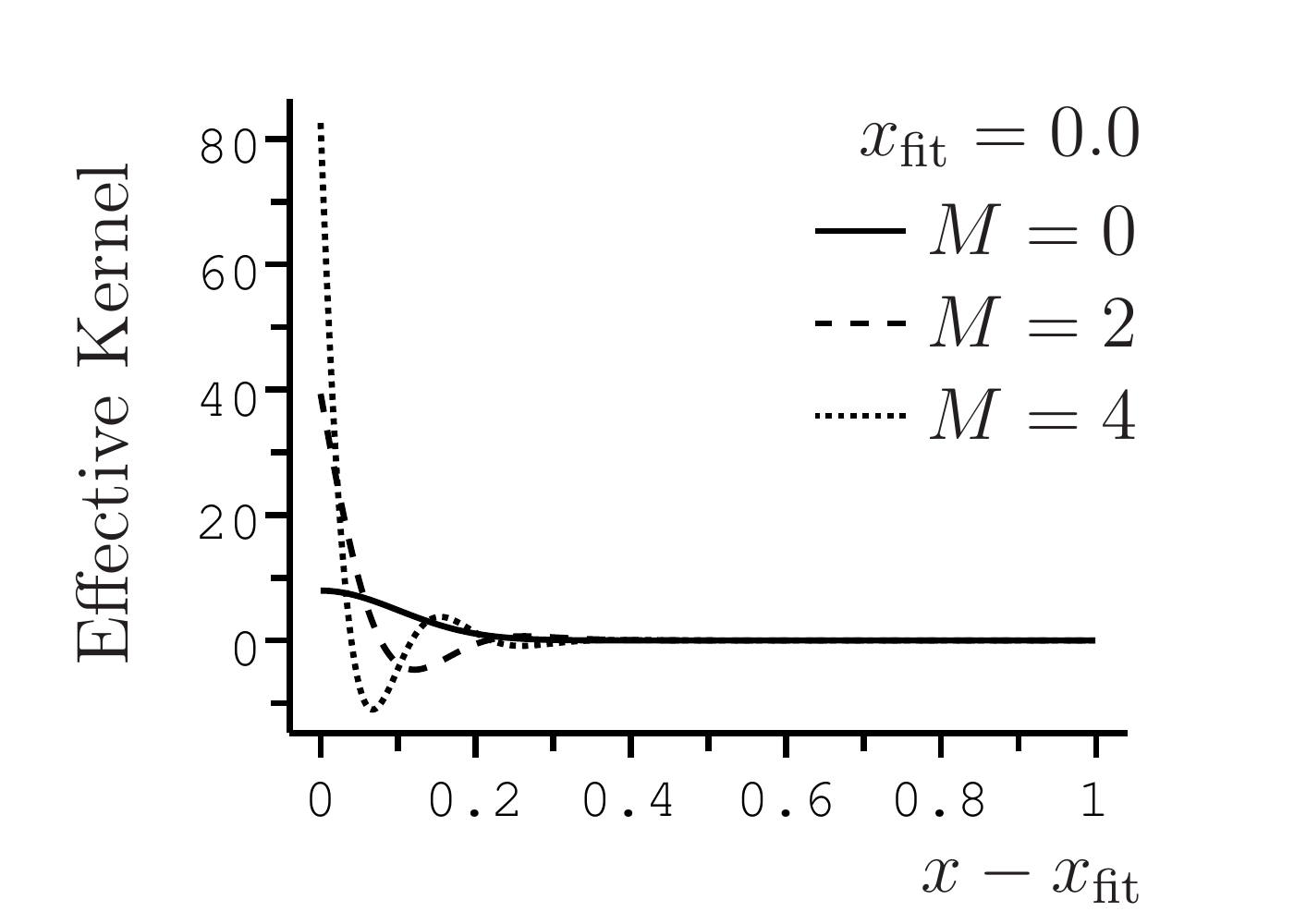} &
\includegraphics[scale=0.36]{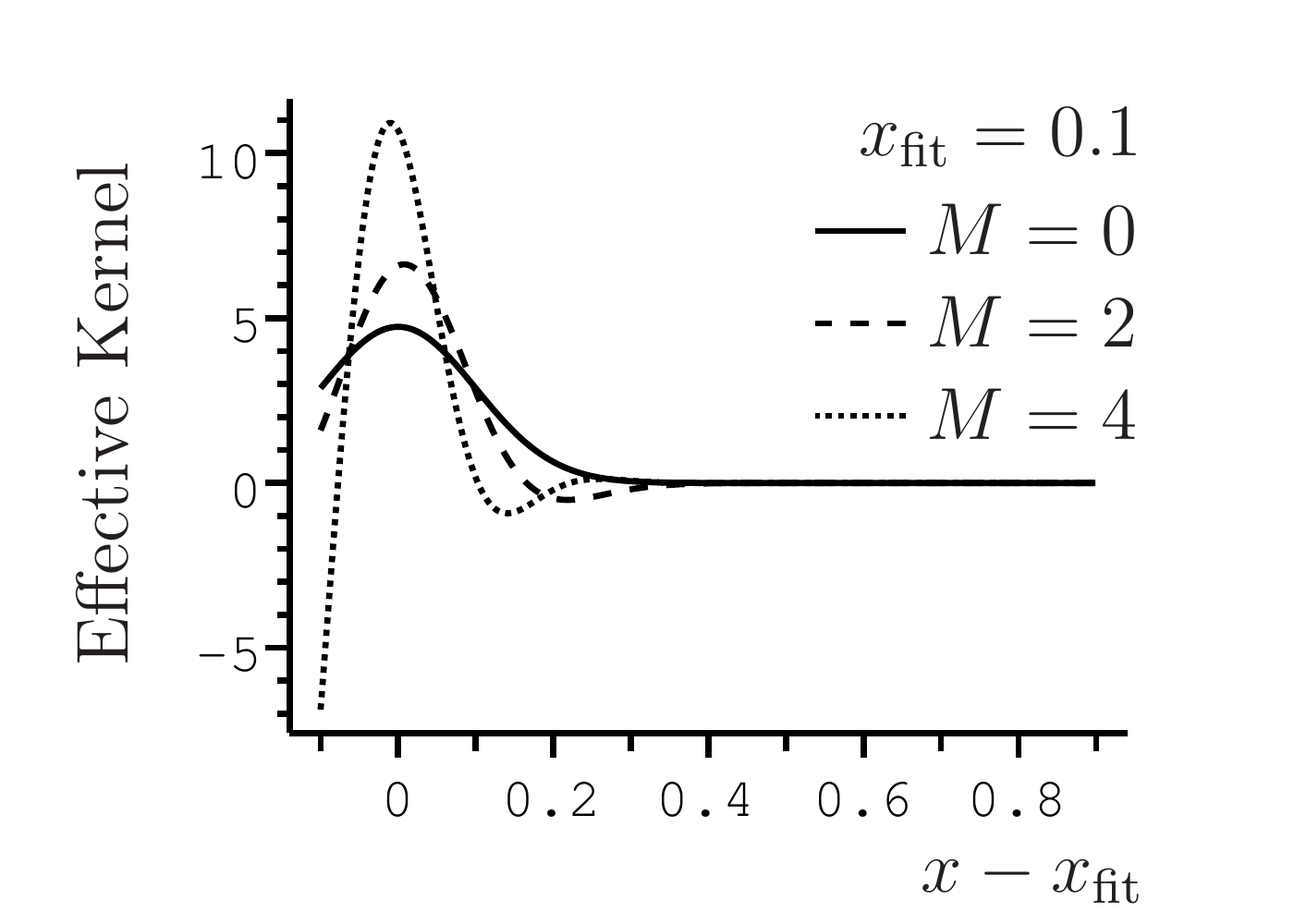} &
\includegraphics[scale=0.36]{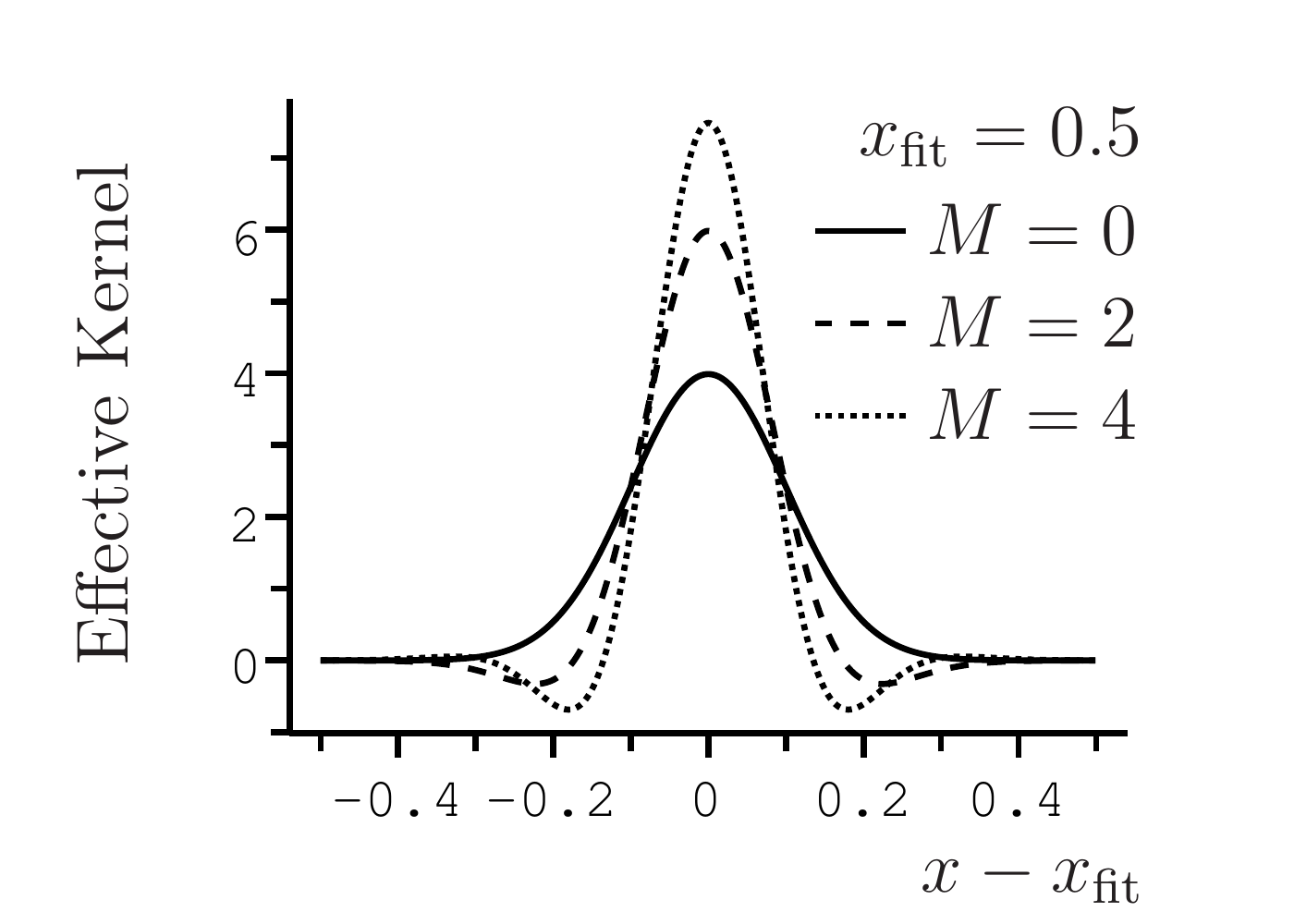} 
\hspace{1cm}
\end{tabular}
\caption{LOrPE effective kernel plots for a~density that is
sharply truncated at $0$, corresponding to support data points: (a) exactly at the boundary,
(b) close to the boundary, and (c) away from the boundary.}
\label{fig:eff-kernel-plots}
\end{center}
\end{figure}
%

\subsection{Orthogonal  series density estimation}

The key idea underlying OSDE for a~univariate density can be traced
back to at least \v{C}encov (1962). Updated monograph-length treatments of
the topic can be found in Tarter \& Lock (1993) and Efromovich (1999).
There is a~strong connection between LOrPE and OSDE. If
    $\{\phi_k\}$ is an orthonormal basis and $f$
    is square integrable, then the classical OSDE of $f(x)$
    is
\begin{equation} \label{classic-osde}
\hfose(x) = \sum_{j=0}^{J}
\theta_j\phi_j(x), \qquad\text{where}\qquad \theta_j=\frac{1}{n}\sum_{i=1}^{n}
\phi_j(x_i).  
\end{equation}
The tuning parameters here consist of the choice of basis functions and their
number, $J$, to carry in the summation. In a~more general form,
and adapted for densities supported on $[a,b]$, this
estimator can be represented as
\begin{equation} \label{general-osde}
\hfose(x) = \frac{1}{b-a}+\sum_{j=1}^{\infty}
w_j\theta_j\phi_j(x),   
\end{equation}
where the $w_j\in [0,1]$ are shrinkage
coefficients (Efromovich, 1999). Comparing (\ref{eq:lorpe}) and (\ref{general-osde}), we
see immediately that
LOrPE can be viewed heuristically as a
localized version of OSDE, since the ``basis functions'' $\{P_k\}$ in the
former are not global, but adjust locally depending on $\xf$. Another facet of the connection between these estimators is revealed in the
following theorem, which establishes that, for large bandwidths, LOrPE is essentially equivalent
to OSDE with a~Legendre polynomial basis.
\begin{theom}\label{th:osde}
In the limit as $h \rightarrow \infty$, the LOrPE estimate
(\ref{eq:expansion}) for pdf $f(x)$ with finite support $[a, b]$,  reduces to
classical OSDE in terms of the basis functions
\[ \phi_j(x)=\sqrt{\frac{2}{b-a}} L_j\left(\frac{2x-a-b}{b-a}\right), \] 
where the $\{L_j\}$ are orthonormal Legendre polynomials on $[-1,1]$.
\end{theom}
\begin{proof}
See the appendix.
\end{proof}

\subsection{Local likelihood density estimation}
In spirit (but
not mathematical detail)
LOrPE is also very similar to LLDE; see Hjort \& Jones (1996), and Loader (1996, 1999). As observed by Loader (1999, ch.~5), LLDE overcomes boundary bias by matching localized sample moments to population
moments using the log-polynomial density approximation (polynomial
approximations on the log scale). As was noted in section~\ref{sec:lorpe-intro}, LOrPE instead
matches localized expectation values of orthogonal polynomials
to their sample values using polynomial density approximations (polynomial
approximations on the original scale). Although
the LLDE approach may be theoretically superior, LOrPE enjoys
the pragmatic advantages of computational speed and numerical
stability, as it does not involve the solution of non-linear
equations  at every grid point.

\section{Selection of Tuning Parameters}\label{sec:mod-select}
This section discusses strategies for selecting the two LOrPE tuning
parameters, bandwidth ($h$) and polynomial degree ($M$). (In principle
the taper function $t(\cdot)$
could also be tuned, but for simplicity we restrict our attention to simple truncation.) We emphasize
this dependence on tuning parameters by writing
\[ \hflpe(x)\equiv\hflpe(x|h,M), \]
and discuss first
an adaptation of the AMISE-optimal plug-in method for KDE
(Silverman's Rule). Methods based on
cross-validation are also proposed. The performance of these
approaches will be examined in section~\ref{sec:simulate}.

\subsection{The plug-in approach}
Note that from Theorem~\ref{th:lorpe-facts}, LOrPE can be viewed as
being equivalent to KDE with a~high-order kernel, $\Keff$. The optimal AMISE expression for KDE with kernel function
$\Keff(\cdot)$ of order $r$, is known to be (e.g., Wand \&
Jones, 1995),
\begin{equation}\label{amise-prof-r}
\text{AMISE}_{h_{\ast}}(r) = \frac{2r+1}{2r}\left[ 2r(r!)^{-2}R\sub{\Keff}(r)^{2r} 
    \mu\sub{\Keff}(r)^2R_{f^{(r)}}(r) n^{-2r} \right]^{1/(2r+1)},
\end{equation} 
with corresponding optimal value of $h$,
\begin{equation}\label{h-prof-r}  
h_{\ast}(r) = \left[ \frac{ (r!)^2R\sub{\Keff}(r) }{
    2rn\mu\sub{\Keff}(r)^2R_{f^{(r)}}(r) } \right]^{1/(2r+1)},  
\end{equation}
where
$f^{(r)}$ denotes the $r$-th derivative of $f$, and
\[
\mu\sub{\Keff}(r) = \int x^r\Keff(x)dx, \qquad R\sub{\Keff}(r)=\int \Keff(x)^2dx, \qquad
R_{f^{(r)}}(r) =  \int f^{(r)}(x)^2dx. 
\]
The unknown moments $\mu\sub{\Keff}(r)$ and
  $R\sub{\Keff}(r)$ can be computed once the underlying
  kernel $K(\cdot)$ is selected; e.g., for Gaussian kernels we have Hermite polynomials, for beta
  kernels Gegenbauer polynomials, etc. For sample sizes in the range $10^2\leq
  n\leq 10^4$,
optimal values of $r$ are likely to be relatively low,
and thus  these
moments can be tabulated across a~few $r$ values with a~symbolic mathematics computer
package, and then included in the relevant programs.
 
The only real difficulty is
estimation of $R_{f^{(r)}}(r)$, but as explained by Wand \& Jones (1995), a~simple
transformation leads to the expression $R_{f^{(r)}}(r)=(-1)^r\psi_{2r}$,
and thus it suffices to study estimation of functionals
$\psi_s\equiv\E[f^{(s)}(X)]$, for $s$ even. For this Wand
\& Jones  (1995) propose
multi-stage direct plug-in algorithms, involving the iteration of a~KDE-type estimator of $\psi_r$ with
optimal bandwidth that depends on $\psi_{s}$, $s>r$. Starting with a
rough estimate of $\psi_s$ at some stage, which can be based on the
well-known value corresponding to a~$N(\mu,\sig^2)$,
this is iterated to arrive at some estimate $\hat{\psi}_s$. A naive estimate of
$\psi_{2r}$ follows by using an estimate of $\sig$ (e.g., sample
standard deviation)\footnote{For the case $r=2$ in the context of KDE this is known as
\emph{Silverman's Rule}.}.
Plugging the resulting estimate of  $R_{f^{(r)}}(r)$ into (\ref{amise-prof-r})
gives eventually,
\begin{equation}\label{norm-amise} 
 \text{AMISE}_{h_{\ast}}(r) \approx \frac{2r+1}{4r\hsig}\left[ \frac{2r(2r!)}{(r!)^3\sqrt{\pi}} 
    \,\mu\sub{\Keff}(r)^2\,\left(\frac{R\sub{\Keff}(r)}{n}\right)^{2r}
  \right]^{1/(2r+1)} .
\end{equation}
Now minimize (\ref{norm-amise}) in $r$ to get
$\hr$ (which by Theorem~\ref{th:lorpe-facts} immediately provides also
an estimate of
$M$). Finally,
substitute $\hr$ into (\ref{h-prof-r}) to obtain the estimates
\begin{equation}\label{amise-optimal-h-M}
 \hh_{AMISE} =  2\hsig\left[ \frac{(\hr!)^3\sqrt{\pi}}{2\hr(2\hr!)n} 
    \, \frac{R\sub{\Keff}(\hr)}{\mu\sub{\Keff}(\hr)^2}
  \right]^{1/(2\hr+1)},\qquad\text{and}\qquad 
\hat{M}_{AMISE}=\begin{cases} \hr+1, & \hr\text{ even},\\ \hr+2, & \hr\text{ odd}.
\end{cases}
\end{equation}  

Of course, this can only serve as a rough estimate, the intent being to provide
reasonable initial values for a more refined search. The fact that
LOrPE naturally self-adjusts near the support end points,
complicates the calculation of the boundary contribution
into the AMISE, as well as the analysis of the bias introduced by the
truncation of the reconstructed density when forced to be non-negative (with subsequent renormalization).

\subsection{Cross-validation methods}
Least squares cross-validation (LSCV) for estimation of a generic
PDF $f$ considers the integrated squared
error of the  density estimate,
\begin{equation}\label{ise-criterion}
ISE = \int{\left[\hf(x) - f(x)\right]^2 dx}. 
\end{equation}
As proposed by Bowman (1984) and Hall (1983), this leads eventually to
minimization of the LSCV criterion. Applied to LOrPE, this yields 
\begin{equation}\label{lscv-criterion}
LSCV(h,M) = \int{\hflpe(x|h,M)^2 dx} - \frac{2}{n}\sum_{i = 1}^{n}\hflpe^{(-i)}(x_{i}|h,M),
\end{equation}
where 
\[
\hflpe^{(-i)}(x|h,M) =\frac{1}{(n-1)h}\sum_{k=0}^{M}\sum_{\substack{j=1\\j\neq i}}^{n}P_{k}\left(\frac{x_{j} - \xf}{h}\right) K\left(\frac{x_{j} - x\sub{fit}}{h}\right) P_{k}\left(\frac{x - \xf}{h}\right),
\]
is the leave-one-out LOrPE density estimate from (\ref{eq:expansion}), obtained by omitting
the $i$-th observation. As suggested in the literature (e.g., Sheather, 2004), the
existence of multiple minima means that it is prudent to plot
$LSCV(h,M)$ over a grid of $h$ and $M$ values. From an asymptotic
perspective, the main drawback of
this criterion is its slow rate of convergence.  

A related simpler and intuitively appealing but less popular approach, is likelihood
cross-validation (LCV), the essential idea dating back to at least
Habbema~\etal{} (1974) and Duin (1976); see for example Silverman
(1986) or Givens \& Hoeting (2013, ch.~10) for an
updated discussion.  This is based on taking the likelihood function of
the leave-one-out density estimate above, leading to minimization of
\[
LCV(h,M) = \prod_{i=1}^{n}\hflpe^{(-i)}(x_{i}|h,M).
\]
Reasoning that the density values at each point
  are taken from slightly different distributions (and not from the same
  distribution as in a genuine likelihood), the term \emph{pseudo}-LCV
  might perhaps be more suitable. 

An obvious obstacle with implementation of this criterion is the
situation when $\hflpe^{(-i)}(x_{i}|h,M)=0$ for some $i$. Its use is also problematic
for densities with infinite support due to the strong influence
exerted by fluctuations in the tails.
To avoid these situations a regularization condition can be introduced, leading to the modified \emph{regularized}
LCV (RLCV) criterion,
\begin{equation}\label{rlcv-criterion}
RLCV(h,M) = \prod_{i=1}^{n}\max\left\{\hflpe^{(-i)}(x_{i}|h,M),
  \frac{\hflpe^{(+i)}(x_i|h,M)}{n^{\alpha}}\right\},
\end{equation}
where $ \alpha>0$ is the regularization parameter, and 
\[
\hflpe^{(+i)}(x|h,M) =\frac{1}{nh}\sum_{k=0}^{M}P_{k}\left(\frac{x_{i} - \xf}{h}\right) K\left(\frac{x_{i} - x\sub{fit}}{h}\right) P_{k}\left(\frac{x - \xf}{h}\right),
\]
is the contribution of data point $x_i$ toward the LOrPE density
estimate (\ref{eq:expansion}). Note therefore that for each $i=1,\ldots,n$ we have
\[ \hflpe(x|h,M)=\hflpe^{(+i)}(x|h,M)+\frac{n-1}{n}\hflpe^{(-i)}(x|h,M). \] 

The case for regularizing LCV was 
made as early as Schuster \& Gregory (1981) who remarked that for tails
exponential and heavier, the use of LCV without regularization
  results in inconsistent density estimates. From a large number of simulations, we have noted that $ \alpha=0.5$
is a reasonable default value. Of course, one can also add $\alpha$ to
the list of tuning parameters to be selected via RLCV, a possibility
that will be explored in section~\ref{sec:simulate}.

\subsection{Effective degrees of freedom and shrinkage}
In situations where truncation of the density below zero
is unnecessary, LOrPE functions as a linear smoother of the EDF,
analogously to KDE. This can be seen by taking the definition of the KDE effective
kernel from Theorem~\ref{th:lorpe-facts}, and observing that we can write (\ref{eq:expansion}) as
\begin{equation*}
\hflpe(x) = \int \frac{1}{h}K\sub{eff}\left(\frac{x-y}{h}\right)\hfemp(y)dy.
\end{equation*}
This suggests the possibility of adapting the idea of \emph{effective degrees of
  freedom} for linear smoothers in a regression setup (Buja~\etal{},
1989), to the analogous situation of density estimation. If $S$ is the
smoothing matrix, the first of three sensible definitions for the effective degrees of
  freedom in a linear smoother, as given by  Buja~\etal{} (1989), is $\Tr(S S^T)$.

For an arbitrary
bandwidth, calculation of this trace appears to be analytically intractable
due to edge effects.
However, in the limit as $h \rightarrow \infty$, recall from Theorem~\ref{th:osde}
that LOrPE converges to OSDE in terms of Legendre polynomials. Now, for a density
fit by a polynomial of degree $M$, the number of degrees of freedom
of the fit (number of free parameters) is obviously
$M$ ($M + 1$ coefficients minus the one constraint from
normalizing the PDF). As the effective degrees of freedom in a smoother
is not limited to integers, this motivates a natural extension of LOrPE
to non-integer values of $M$. Through suitable
choice of the taper function, we can ensure that the effective degrees
of  freedom in any given fit is always $M$. 

To formalize this, consider without loss of generality a PDF supported on $[-1, 1]$.
With $t(\cdot)$ a chosen taper function and the $\{L_k\}$ defined as
in Theorem~\ref{th:osde}, OSDE smoothing is then seen to be
performed by the linear operator $S(x, y) = \sum_{k=0}^{\infty}
t(k)L_k(x) L_k(y)$, in an appropriate inner product space. Requiring
the inner product with the EDF to yield OSDE, motivates the following definition:
\begin{equation*}
\hfose(x)=\langle S(x, y),\hfemp(y)\rangle\equiv\int_{-1}^1 S(x, y) \hfemp(y)dy = 
\frac{1}{n}\sum_{i=1}^{n}\sum_{k=0}^{\infty}t(k)L_{k}\left(x\right)L_{k}\left(x_i\right).
\end{equation*}
This is now in the form of (\ref{general-osde}), with the $t(k)$
playing the role of the shrinkage coefficients $w_k$.
The operator $S$ is in fact self-adjoint (symmetric), so that
\begin{eqnarray*} 
S S^T &=& \langle S(x, z),S(z,y)\rangle = \int_{-1}^1 S(x, z) S(z, y)
dz \\
&=& \int_{-1}^1 \sum_{k=0}^{\infty} t(k) L_k(x) L_k(z) \sum_{j=0}^{\infty} t(j) L_j(z) L_j(y) dz\\
&=& \sum_{k=0}^{\infty} t^2(k) L_k(x) L_k(y),
\end{eqnarray*}
the last line following from identity
(\ref{legendre-ortho}). Additionally, note that we have
\begin{equation*}
\Tr(S S^T) = \int_{-1}^1 S S^T \delta(x - y) dx dy = \int_{-1}^1\sum_{k=0}^{\infty}
t^2(k)L_k^2(x)dx = \sum_{k=0}^{\infty}t^2(k).
\end{equation*}
Adapting the above definition for the effective degrees of
  freedom from Buja~\etal{} (1989) to density estimation, we therefore
  arrive at the identity
\begin{equation}\label{eff-df-identity}
M=\Tr(S S^T)-1=\sum_{k=0}^{\infty}t^2(k)-1. 
\end{equation}
There are many possible choices for $t(\cdot)$ which would make
(\ref{eff-df-identity}) work, but perhaps the simplest is to take the
step function approach of section~\ref{sec:lorpe-intro}. However, if
the optimal $M$ is not an integer, an extra adjustment is needed, so
that a more general prescription (with $m=\lfloor M\rfloor$ denoting
the largest integer less than or equal to $M$) is to define:
\[
 t(k) = \begin{cases}
1, & k\leq m, \\
\sqrt{M-m}, & k=m+1,\\
0, & k\geq m+2. 
\end{cases}
\]
Throughout the paper, we adopt these shrinkage coefficients in all
instances where LOrPE is applied, for any given bandwidth $h$. 


\section{Simulations}\label{sec:simulate}

The primary goal of this section is to compare the MISE performance of LOrPE
with that of its main competitor, KDE. This will be done both from
oracle and non-oracle based perspectives. The oracle based
comparisons, so called
because the optimization has access to the true analytical ISE, are aimed at benchmarking the performance of the two
methods, especially with regard to estimating densities that are sharply truncated. 
The non-oracle based comparisons will explore the performance of LOrPE
to all of its rivals and analogues discussed thus far: KDE, OSDE, and LLDE.

To produce a spanning set of densities $f$ to be investigated,
some elements from the
list in Wand \& Jones (1995, Table 2.2) were employed as a starting
point. This includes the KDE-optimal Beta(4,4) on $(-1,1)$, as derived
by Terrell (1990) for minimizing AMISE through minimization of total curvature.
To these were added a few that are sharply truncated. Table~\ref{tab:dist-list} lists the choice of distributions
selected for the simulation study, where $\phi(z)$ and $\Phi(z)$
denote the PDF and CDF of a standard normal. In particular, there are
three distributions with sharp boundaries: two standard normals, one truncated
at $0$ and the other at $-1$, and a standard exponential. It is
expected that KDE will handle the N(0,1) truncated at $0$ well using 
data reflection (or mirroring), and it would therefore be interesting to compare its
performance with that of LOrPE which does not enjoy this
advantage.  On the other
hand, we would expect to see LOrPE  outperform KDE for the
N(0,1) truncated at $-1$, as the data reflecting method doesn't work
well in this case (due to discontinuity of the first derivative).

\begin{table}[htb]
\caption{List of distributions for simulations.}
\label{tab:dist-list}
\begin{center}
\begin{tabular}{lr}
Name of $X$ &  Distribution/Density of $X$ \\ 
\hline\hline
Beta(4,4) on $[-1,1]$ & $f^*(x)=\frac{35}{32}(1-x^2)^3I(|x|<1)$ \\
(optimal by KDE) &  \\
\hline
N(0,1) \nc & $f(x)=\phi(x)$   \\ 
\hline
Normal Mix 1 &  $X\sim\frac{3}{4}Z_1+\frac{1}{4}Z_2$  \\ 
(bimodal)& $Z_1\sim\text{N}(0,1)$ and $Z_2\sim\text{N}(3/2,1/9)$  \\ 
\hline
Exponential(1) & $f(x)=e^{-x}I(x>0)$ \\
(sharp boundary at $0$) & \\
\hline
N(0,1) on $[0,\infty)$ \nc & $f(x)=2\phi(x)I(x>0)$  \\
(truncated at $0$)&  \\
\hline
N(0,1) on $[-1,\infty)$ \nc & $f(x)=\frac{1}{\Phi(1)}\phi(x)I(x>-1)$ \\
(truncated at $-1$)& \\
\hline
Normal Mix 2 &  $X\sim\frac{2}{3}Z_1+\frac{1}{3}Z_2$  \\
(sharp peak at $0$) & $Z_1\sim\text{N}(0,1)$ and $Z_2\sim\text{N}(0,1/100)$  \\ 
\hline\hline
\end{tabular}
\end{center}
\end{table}

\subsection{Oracle MISE comparisons: LOrPE vs.~KDE}\label{subsec:oracle-mise}
The ``oracle'' MISE comparisons, called ``best case'' by Jones \& Henderson (2007), are useful for
benchmarking LOrPE vs.~KDE in determining the
best possible performance for each method with regard to estimation of
a particular density. Dassanayake (2014) details the procedure used to effect these comparisons
for each of the distributions in
Table~\ref{tab:dist-list}. This involves performing a computationally
intensive search for the  optimal $h^{\ast}$ and $M^{\ast}$ that
minimize the MISE over grids of polynomial degree values,
  $M\in\mathcal{M}$, and  bandwidths
  $h\in\mathcal{H}$. MISEs were
calculated by averaging 1,000 numerical estimates of ISE values  (\ref{ise-criterion}). For KDE, $M$
  is related to the kernel order via result (iii) of
  Theorem~\ref{th:lorpe-facts}, and is therefore the approximate kernel
  order. 

Figure~\ref{fig:oracle-mise}
displays the resulting $\log_{10}(\text{MISE}(h^{\ast},M))$ values as a function of
$M$, for each of LOrPE and KDE, and sample sizes of $n=10^3$
 and $n=10^5$.
%
\begin{figure}[tbh]
\begin{center}
\includegraphics[scale=1]{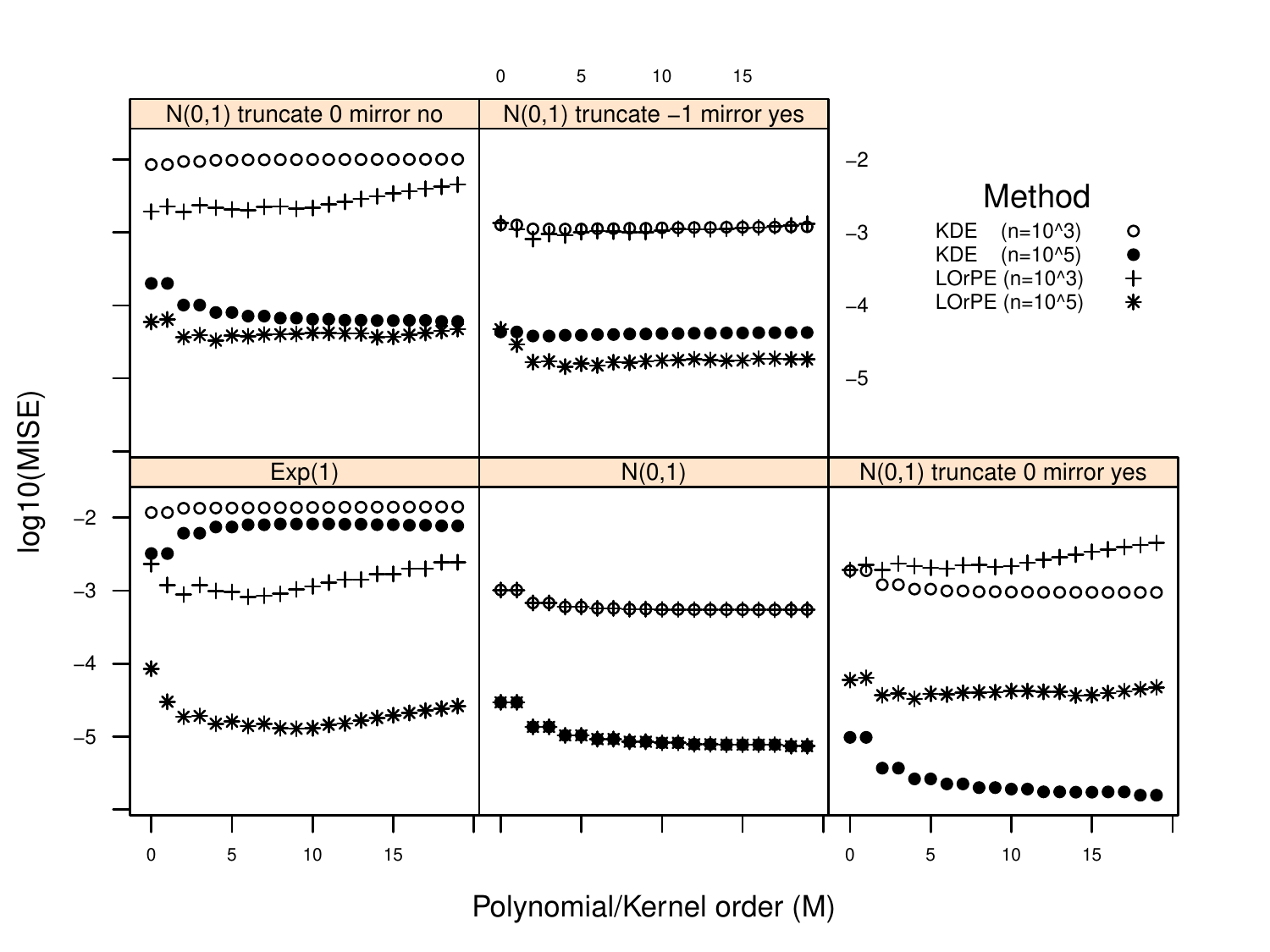}
\caption{Oracle MISE values for sample sizes
  $n=10^3$ and $n=10^5$, as a function of polynomial order $M$ (LOrPE)
  or approximate kernel order
  $M$ (KDE).}
\label{fig:oracle-mise}
\end{center}
\end{figure}
%
According to these graphical summaries, it is clear that LOrPE works in
a similar manner to KDE when estimating densities with exponentially
declining tails at both ends of the support,
such as the $N(0,1)$. Similar results were observed for the Beta$(4,4)$, and
the two Normal Mixes (not shown). For densities with sharp edges,  LOrPE tends to
attain lower MISE values than KDE. The $N(0,1)$ truncated 
at $0$ (with KDE mirroring) is a notable exception; but the better
performance of KDE is only really discernible at larger sample sizes and
higher kernel orders. If the crucial data mirroring property of KDE at the
boundaries is removed, then the tables are reversed in
favor of LOrPE, particularly at small sample sizes and low kernel
orders. The Exponential$(1)$ constitutes a dramatic case in favor of
LOrPE, while the  $N(0,1)$ truncated at $-1$ (with KDE benefiting from
mirroring) is somewhere in  between these
two extremes. Note that LOrPE does not use data mirroring (although it
can use kernel
mirroring whereby the weight function is reflected at the boundary and added
to the non-reflected part).

The appropriate minimum oracle $\log_{10}(\text{MISE}(h^{\ast},M^{\ast}))$ values for all the
densities of
  Table~\ref{tab:dist-list}, are displayed in
  Table~\ref{tab:oracle-mises}, along with the corresponding optimal
  $(M^{\ast},h^{\ast})$. Note that all truncated $N(0,1)$ KDE values were obtained using
  data mirroring, whereas the un-truncated $N(0,1)$ did not. As can be
  seen, at lower sample sizes all LOrPE estimates have lower (or the same) MISE,
  except for the truncated normals. However, this KDE advantage for the
  truncated normal at $-1$ gradually erodes, so that at higher sample
  sizes only the KDE estimates for the $0$ truncated $N(0,1)$ persist
  in having  lower MISE than LOrPE.
\begin{table}[tbh]
\caption{Oracle $log_{10}(\text{MISE})$ values for the densities in
  Table~\ref{tab:dist-list} as a function of sample size. The values
  in parentheses correspond to the optimal $(M^{\ast},h^{\ast})$,
  where $M$ is the polynomial order (LOrPE) or approximate kernel order (KDE) and
  $h$ is the bandwidth. For each density and each $n$, the lowest of the two MISE values
appears in bold face.}%
\centering
\label{tab:oracle-mises}
\begin{tabular}{lcccccc}              
\multirow{2}{*}    & \multicolumn{2}{c}{\underline{$n=10^2$}} &
\multicolumn{2}{c}{\underline{$n=10^3$}} & \multicolumn{2}{c}{\underline{$n=10^4$}} \\
{Distribution}     & LOrPE   & KDE & LOrPE & KDE & LOrPE & KDE	\\  
\hline \hline
\multirow{2}{*}{N(0,1)}  
& \bf{-2.441} & -2.440 & \bf{-3.260} & \bf{-3.260} & \bf{-4.177} & \bf{-4.177}  \\
& (19, 11.4) & (13, 8.3) & (16, 8.2) & (17, 8.2) & (17, 7.0) & (17, 7.0) \\ [3pt]
\multirow{2}{*}{Normal Mix 1}  
& \bf{-2.014} & \bf{-2.014} & \bf{-2.774} & \bf{-2.774} & \bf{-3.661} & \bf{-3.661} \\
& (0, 1.0)  & (0, 1.0)  & (4, 1.5)  & (4, 1.5)  & (11, 2.0) & (11, 2.0) \\ [3pt]
\multirow{2}{*}{Normal Mix 2 }  
& \bf{-1.378} & -1.370 & \bf{-2.208} & \bf{-2.208} & \bf{-3.108}  & \bf{-3.108} \\
& (0, 0.25) & (2, 0.43) & (6, 0.51) & (6, 0.51) & (14, 0.74) &(14, 0.74) \\ [3pt]
\multirow{2}{*}{N(0,1) on $[0,\infty)$}  
& -2.243 & \bf{-2.475} & -2.997 & \bf{-3.318} & -3.869 & \bf{-4.213}	\\
& (0, 1.2) & (17, 9.7)& (2, 1.9) & (19, 8.6) & (4, 2.7) & (18, 7.5) \\  [3pt]
\multirow{2}{*}{N(0,1) on $[-1,\infty)$}  
& -2.223 & \bf{-2.241} & \bf{-3.091} & -2.609 & \bf{-3.932} & -3.010	\\
& (2, 3.0) & (4, 3.9) & (2, 2.1) & (1, 0.79) & (2, 1.6) & (0, 0.19) \\ [3pt]
\multirow{2}{*}{Beta(4,4)}  
& -2.044 & \bf{-2.064} & \bf{-2.890} & -2.886 & \bf{-3.824} & -3.705	\\
& (4, 13.0) & (8, 2.04)& (4, 1.5) & (10, 2.9) & (6, 11.6) & (9, 1.5) \\ [3pt]
\multirow{2}{*}{Exponential(1)}  
& \bf{-2.265} & -1.462 & \bf{-3.085} & -1.954 & \bf{-4.002} & -2.392	\\
& (2, 4.1) & (0, 0.48) & (6, 13.2) & (0, 0.16) & (8, 13.7) & (0, 0.082) \\
\hline \hline
\end{tabular}
\end{table}

\subsection{Non-oracle MISE comparisons}\label{subsec:non-oracle-mise}
The intent in this section is to compare LOrPE MISE values to those of its closest
competitors, KDE, LLDE, and OSDE, in a realistic (non-oracle) setting. In order to make these comparisons
as fair as possible in terms of mimicking an unsophisticated user, ``reasonable'' default settings were used for the
the respective tuning parameters of each method. The details are as
follows.
\begin{description}
\item[LOrPE:] Uses the plug-in estimates from
  (\ref{amise-optimal-h-M}), implemented via the NPStat package (Volobouev, 2012).
\item[KDE:] Uses the Sheather \& Jones (1991)  two-stage plug-in ("dpi" or
"direct plug-in")  bandwidth with
  a normal kernel and sample standard deviation as the
  estimate of scale, implemented via R library \texttt{ks}. 
\item[LLDE:] Uses the above KDE plug-in bandwidth, a Gaussian kernel,
  and zero-order  polynomial, implemented via the R library \texttt{locfit}.
\item[OSDE:] The estimator in
  (\ref{classic-osde}) was coded with the number of terms, $J$, chosen
          according to the Hart (1985) scheme. The NPStat
package (Volobouev, 2012) is used to generate the necessary orthogonal
          polynomials on a grid (consisting of 2,048 points). The lowest and highest order statistics from the sample of
size $n$ are mapped to the $1/(2n)$ and $1 -
1/(2n)$ quantiles, respectively. All other points are then mapped linearly using these two extremes. The support of the density
is now estimated by inversely mapping the $[0, 1]$
interval. The discrete
analog of Legendre polynomials are employed; generated by the Gram-Schmidt
procedure for a uniform weight on the grid in $[0, 1]$.
\end{description}

MISEs were calculated empirically as in section~\ref{subsec:oracle-mise}. The data were once again simulated
from most of the distributions in Table~\ref{tab:dist-list}, as well
as  Student's $t$ with 1, 2, and 3 degrees of freedom truncated to the
interval $[-1,2]$. The results are presented on
Table~\ref{tab:non-oracle-mises} which summarizes the
$log_{10}(\text{MISE})$  values for three different sample sizes
within each distribution. We note that  LOrPE yields consistently minimum
MISE values for the sharply truncated normal distributions and the
Exponential.  For the truncated $t$ distributions the results are
mixed, but LOrPE tends to dominate for larger sample sizes. In nearly
all cases where LOrPE does not yield the minimum
MISE, it is a close second.
%
\begin{table}[tbh]
\caption{Non-oracle $log_{10}(\text{MISE})$ values for 4 estimators of
  the true density. MISEs are  based on 1,000
  realizations simulated from a variety of distributions and sample
  sizes ($n$).  For each
  distribution and each $n$, the lowest of the 4 MISE values appears in bold face.}%
\centering
\label{tab:non-oracle-mises}
\begin{tabular}{lccccc}              
Distribution & $log_{10}(n)$  & LOrPE   & KDE & LLDE & OSDE \\  
\hline \hline
\multirow{3}{*}{N(0,1)}  
& $2$                 & -2.138 & \bf{-2.198} & -2.183 & -1.634 \\ 
& $3$                 & -3.088 & -2.973 & \bf{-3.179} & -1.650 \\ 
& $4$                 & -4.045 & -3.741 & \bf{-4.158} & -2.652 \\ 
\hline
\multirow{3}{*}{N(0,1) on $[0,\infty)$} 
& $2$                 & \bf{-2.177} & -1.576 & -1.427 & -1.666 \\ 
& $3$                 & \bf{-2.923} & -2.010 & -1.594 & -2.642 \\ 
& $4$                 & \bf{-3.770} & -2.392 & -1.613 & -3.634 \\ 
\hline
\multirow{3}{*}{N(0,1) on $[-1,\infty)$} 
& $2$                 & \bf{-2.085} & -2.023 & -1.837 & -1.823 \\ 
& $3$                 & \bf{-3.005} & -2.564 & -2.188 & -2.799 \\ 
& $4$                 & \bf{-3.874} & -2.980 & -2.248 & -3.776 \\ 
\hline
\multirow{3}{*}{Normal Mix 1} 
& $2$                 & -1.743 & \bf{-1.888} & -1.824 & -1.148 \\ 
& $3$                 & -2.108 & \bf{-2.223} & -2.028 & -1.149 \\ 
& $4$                 & \bf{-2.477} & -2.278 & -2.060 & -1.149 \\ 
\hline
\multirow{3}{*}{Exponential(1)} 
& $2$                 & \bf{-2.239} & -1.374 & -1.299 & -0.677 \\ 
& $3$                 & \bf{-2.915} & -1.783 & -1.386 & -1.2328 \\ 
& $4$                 & \bf{-3.740} & -2.157 & -1.393 & -0.679 \\ 
\hline
\multirow{3}{*}{$t(1)$ on $[-1,2]$} 
& $2$                 & -1.891 & -2.317 & \bf{-2.447} & -1.347 \\
& $3$                 & -2.712 & -2.694 & \bf{-3.000} & -2.337 \\ 
& $4$                 & \bf{-3.661} & -3.118 & -3.128 & -3.337 \\ 
\hline
\multirow{3}{*}{$t(2)$ on $[-1,2]$} 
& $2$                 & -1.980 & -2.346 & \bf{-2.366} & -1.408 \\ 
& $3$                 & -2.839 & -3.065 & \bf{-3.191} & -2.404 \\ 
& $4$                 & \bf{-3.724} & -3.546 & -3.591 & -3.400 \\ 
\hline
\multirow{3}{*}{$t(3)$ on $[-1,2]$} 
& $2$                 & -2.039 & \bf{-2.328} & -2.289 & -1.437 \\ 
& $3$                 & -2.879 & -3.061 & \bf{-3.207} & -2.427 \\ 
& $4$                 & -3.769 & \bf{-3.856} & -3.763 & -3.416 \\ 
\hline \hline 
\end{tabular}
\end{table}

\subsection{Oracle and non-oracle MISE comparisons: LOrPE vs.~KDE}\label{subsec:oracle-nonoracle-mise}

Recall that the LOrPE plug-in approach is meant to serve as an initial estimate
in a more refined search for appropriate $h$ and $M$ values.  Since
plug-in formulae do not take boundary effects into account, we would expect  sub-optimal
performance from LOrPE in regard to estimation in the
vicinity of the support boundary. The already good LOrPE plug-in performance seen in
section~\ref{subsec:non-oracle-mise} could therefore potentially be 
improved by using cross-validation methods.  Given that
oracle comparisons provide lower bounds on MISE values, we may ask two interesting
questions of LOrPE
cross-validation methods: (i) how close can they get to LOrPE  oracle values, and (ii) how close can they get to KDE oracle values. 

This section aims to answer these questions, using both the LSCV
and RLCV criteria, as described by equations (\ref{lscv-criterion})
and (\ref{rlcv-criterion}), respectively, with the regularization
parameter set at $\al=0.5$ in the latter. Both oracle and
non-oracle methods are considered, and as such the simulation details for the
former parallel those of section~\ref{subsec:oracle-mise}, while those
for the latter are
identical to section~\ref{subsec:non-oracle-mise}. For KDE oracle
computations: the $N(0,1)$ and $N(0,1)$ truncated at $-1$ did not use
data mirroring, while the $N(0,1)$
truncated at $0$ used mirroring. This time a
variety of sample sizes were considered in order to reveal any
possible convergence of methods as $n\rightarrow\infty$.
 Also, for brevity only 6 of the (representative) distributions listed in
  Table~\ref{tab:non-oracle-mises} were examined. 

The resulting
  $log_{10}(\text{MISE})$ values appear plotted vs.~sample size in
  Figure~\ref{fig:non-oracle-mise}. The answer to the above two
  questions seems clear. First,  LOrPE
cross-validation methods come very close to LOrPE  oracle values, with
the RLCV criterion dominating LSCV most of the time. Secondly, and
remarkably, except for the $N(0,1)$ and $0$ truncated $N(0,1)$, LOrPE
cross-validation methods produce consistently lower MISE values than KDE oracle.  

\begin{figure}[tbh]
\begin{center}
\includegraphics[scale=1]{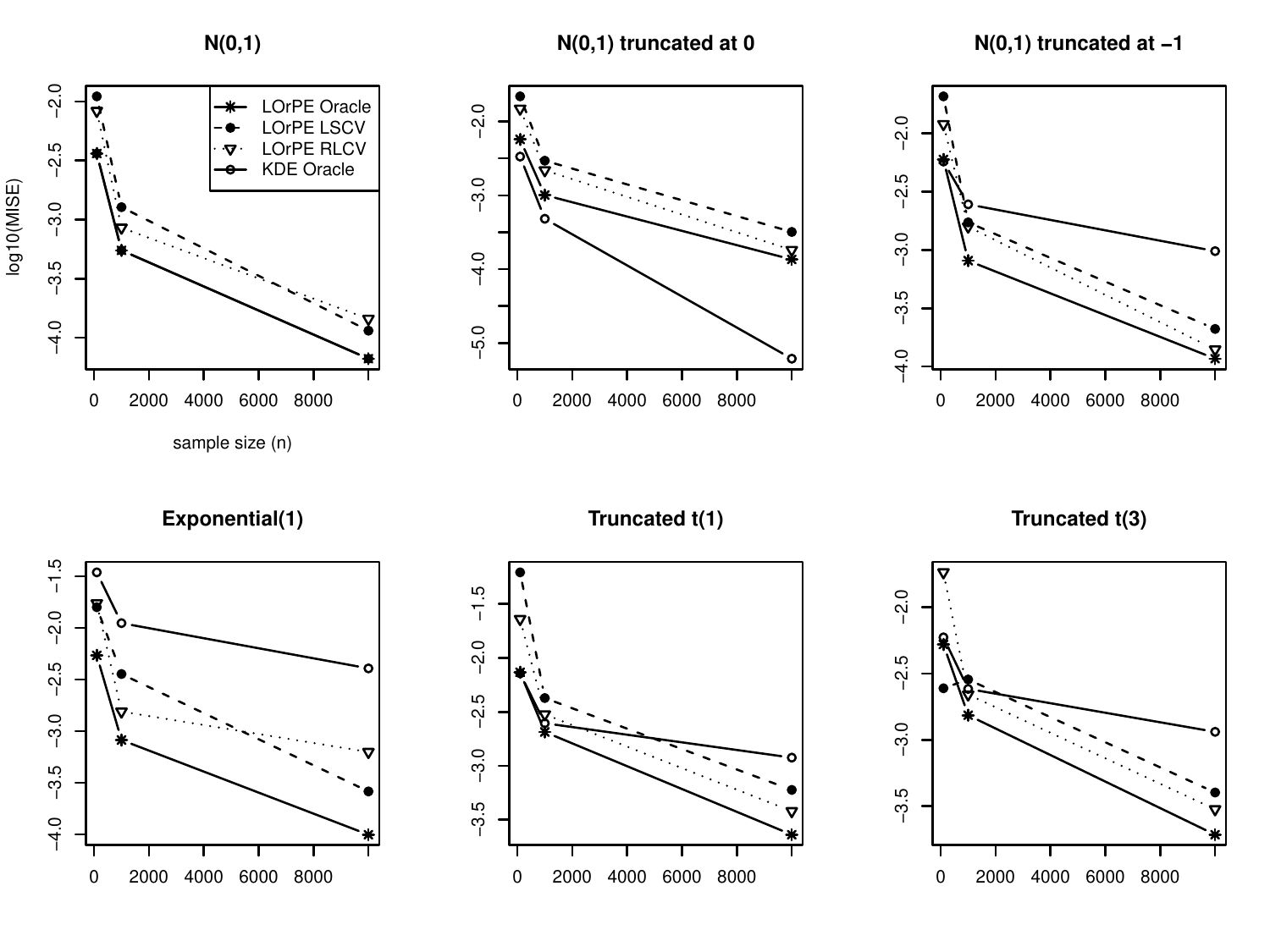}
\caption{Plots of oracle (solid lines) and non-oracle $log_{10}(\text{MISE})$ values
  for LOrPE and KDE. The non-oracle methods of LSCV (dashed lines) and RLCV (dotted lines) apply only to  LOrPE.}
\label{fig:non-oracle-mise}
\end{center}
\end{figure}

In some cases, and especially at small sample sizes, the LOrPE-RLCV 
method may not be achieving the lowest possible
MISE. One reason for this could be that the regularization parameter
choice of $\al=0.5$ is not optimal. To investigate this issue,
Figure~\ref{fig:r-mise} plots the $log_{10}(\text{MISE})$ values
vs.~$\al\in [0,1]$ for the distributions considered in
Figure~\ref{fig:non-oracle-mise}, and for sample size $n=10^3$ only. The error bars around each value
extend from the $84.13^{th}$  to the $15.87^{th}$ percentiles divided
by $2\sqrt{n}$, and provide a sense of sampling variability through a
robust measure of the standard error. It is clear that, perhaps with
the exception of the $N(0,1)$ case, LOrPE-RLCV  is reasonably
insensitive to the choice of $\al$. This suggests that it may not
be necessary to estimate this extra tuning parameter, and just use a
default value of $\al=0.5$. 

\begin{figure}[tbh]
\begin{center}
\includegraphics[scale=1]{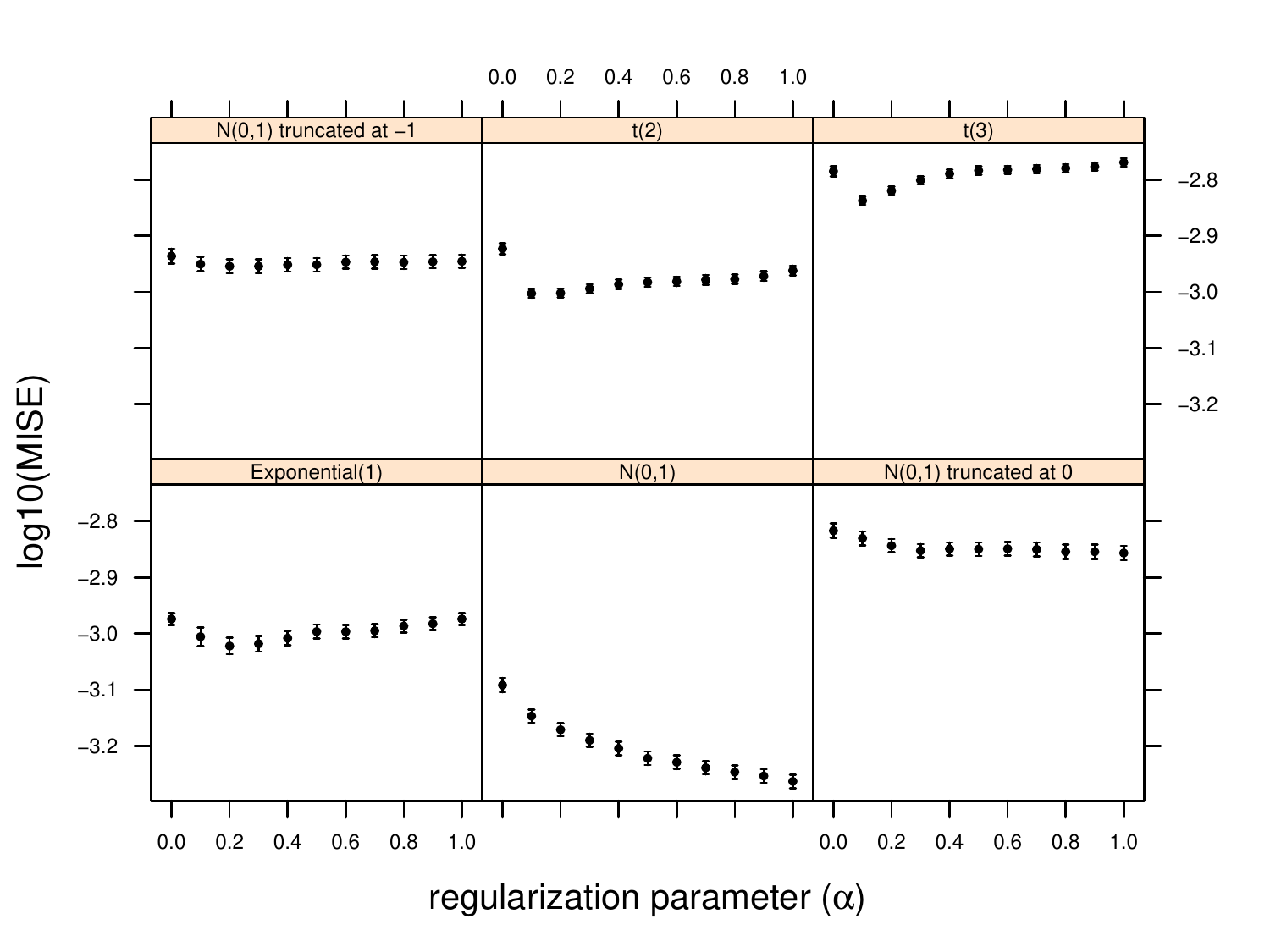}
\caption{Plots of $log_{10}(\text{MISE})$ values vs.~the
  regularization parameter $\al$ for the LOrPE-RLCV method applied to
  1,000 simulated datasets of sample
  size $n=10^3$. The error bars provide a robust measure of the standard error.}
\label{fig:r-mise}
\end{center}
\end{figure}

\section{Real Data Application}\label{sec:real-data}

As an illustration of the proposed methodology, we consider the
lengths of $n=86$ spells of psychiatric treatment (days) undergone by patients
used as controls in a study of suicide risks (Copas \& Fryer,
1980). The data were presented by Silverman (1986, Table 2.1), who used
them to demonstrate certain inadequacies with KDE. Scaled to the
unit interval by dividing all observations by the largest value of
737, it is publicly available in the R library \texttt{bde} as 
"suicide.r''. Figure~\ref{fig:suicide-data} displays a histogram with
rugplot, and five density estimates. Sturges' formula is used to compute the breaks and number of classes in the
histogram shaded in gray (the default in R function ``hist'').

\begin{figure}[tbh]
\begin{center}
\includegraphics[scale=1]{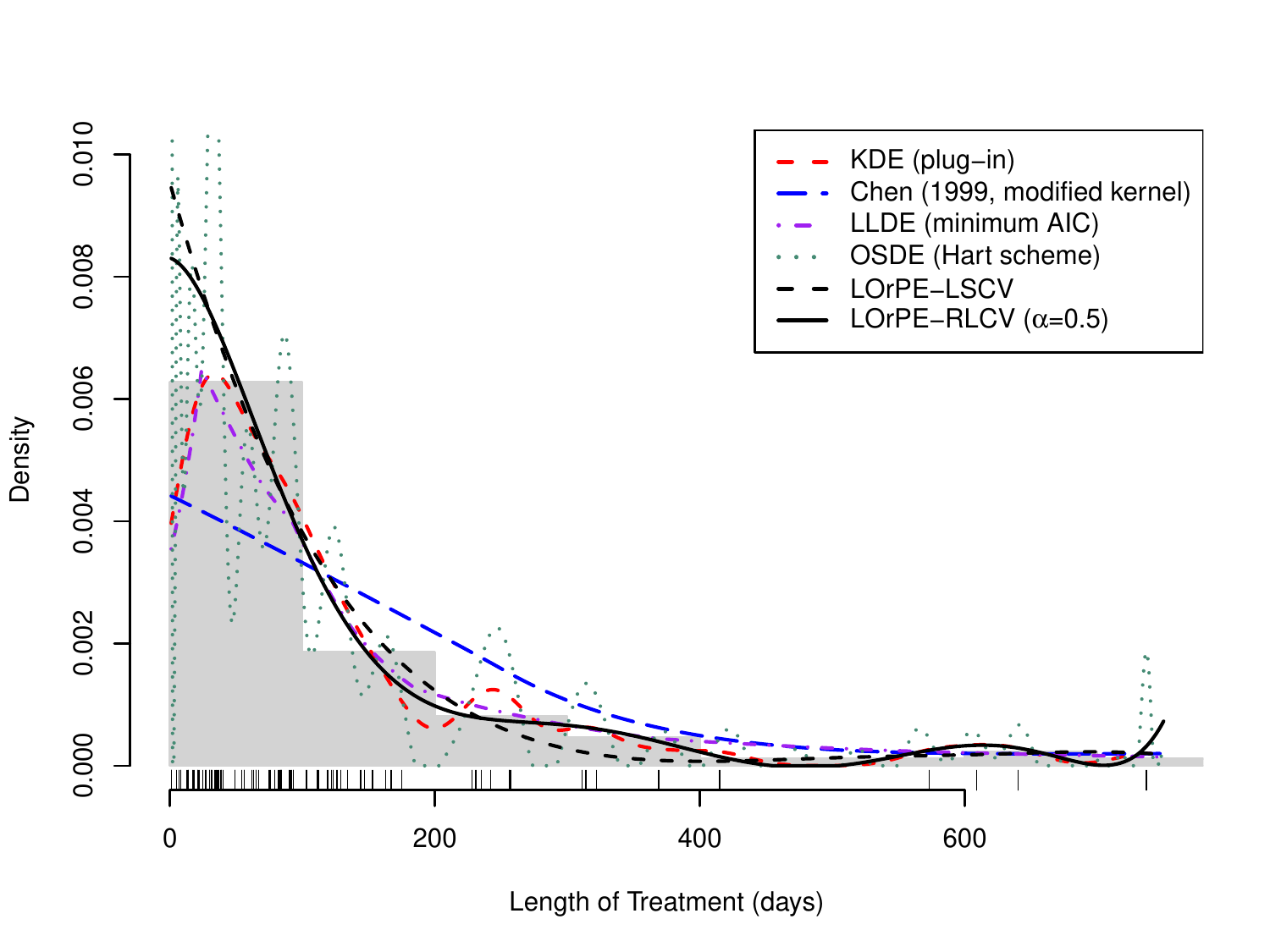}
\caption{Density estimates for the Suicide data: KDE (red/dashed), Chen
  (1999) (blue/longdash), LLDE (purple/dotdash), OSDE (green/dotted), LOrPE-LSCV (black/dashed), LOrPE-RLCV (black/solid).}
\label{fig:suicide-data}
\end{center}
\end{figure}

KDE (red dashed lines) uses the plug-in
bandwidth as described in section~\ref{subsec:non-oracle-mise}. As expected, there is an apparent bias at the
  left end of the support, the estimate dips down toward zero, whereas
  the data suggests there should be a large amount of mass in that
  vicinity. A similar outcome occurs with LLDE (purple dotdash lines), which displays less
  ``wigglyness'' in the tail, but a sharp ``kink'' at the peak. As
  suggested by Loader (1999), greater care was exercised in selecting
  appropriate values for the LLDE tuning parameters: we used AIC to
  identify  the optimal nearest neighbor component of the smoothing
  parameter and polynomial order, instead of the (quicker) KDE plug-in
  bandwidth and degree zero of
  section~\ref{subsec:non-oracle-mise}, as a means of specifying the effective
  degrees of freedom. No appreciable changes were observed with
  kernels different from  Gaussian. 
OSDE (green dotted lines) obviously undersmoothes badly, a consequence of the degree
preferred by Hart's (1985) method being $J=48$. 

Rather more believable performance was obtained with Chen's (1999)
boundary-corrected beta-kernel density estimator (blue longdash lines),
which picks up the mass at the peak, but seems to be somewhat
oversmoothed. This is Chen's (1999) second beta-kernel estimator,
called ``modified'' in the R library \texttt{bde} with which it is
implemented, since Chen (1999) showed it consistently outperforms the 
first beta-kernel estimator. The critical bandwidth tuning parameter
is set at the default value of $b=n^{-2/5}$, the AMISE optimal order
for such kernels (Chen, 1999). Finally, we note the arguably superior
performance of LOrPE (black dashed and solid lines). LORPE-RLCV uses the default value of
$\al=0.5$ for the regularization parameter, as suggested by the
simulations in section~\ref{subsec:oracle-nonoracle-mise}, and the
optimal degree and bandwidth were $M=7$ and $h=2047.2$. LORPE-LSCV
delivers a similar performance with $M=2.9$ and $h=605.3$.

\section{Summary Remarks}
We have shown that LOrPE is a useful extension to the (already vast) array of tools for
nonparametric density estimation. This novel idea has at its basis the
local expansion of the EDF
into a series of orthogonal polynomials around a selection of
grid points. It was demonstrated that away from the support boundary
LOrPE essentially functions like KDE with a high-order kernel, whereas close
to the boundary LOrPE is adaptive in the sense that its effective
kernels naturally change shape to accommodate the endpoint, thereby
reducing boundary bias. Faster asymptotic convergence rates follow
naturally by virtue
of the higher-order kernels. LOrPE also shares important connections with LLDE and OSDE. 
Simulations demonstrated that LOrPE generally outperforms these estimators, and especially KDE, when
estimating densities with sharp boundaries. Also, LOrPE allows for
the inclusion of a taper function, a feature which takes LOrPE beyond
KDE with high-order kernels. 

These reasons make LOrPE applicable in a wider range of problems than
KDE. When estimating distributions which decay rapidly at infinity,
LOrPE results are identical to KDE. Additionally, the local
polynomial modeling can effectively reduce the bias
for densities with several (at least $M$) continuous
derivatives. A proper balance of $h$ and $M$ can thus result in a
better overall estimator. Cross-validation, and especially a
regularized version of likelihood cross-validation, seems to be a promising way
of selecting appropriate values for these tuning parameters. For large
$n$, the simulations suggest LOrPE MISE approaches the oracle (or ``best
case'') MISE. Finally, LOrPE calculations remain essentially unchanged in multivariate settings, requiring only a switch to multivariate orthogonal polynomial systems.


\appendix

\section{Proof of Theorem~\ref{th:lorpe-facts}}
Substituting the expression for $c_k(\cdot)$ from~(\ref{eq:emp-ck})
into (\ref{eq:lorpe}) gives 
\begin{eqnarray*}
\hflpe(x) &=& \sum_{k=0}^{\infty} t(k) c_{k}(x\sub{fit}, h)
P_{k}\left(\frac{x - x\sub{fit}}{h}\right) \\
&=& \sum_{k=0}^{\infty} t(k) \left\{ \frac{1}{nh}\sum_{i=1}^{n}P_{k}\left(\frac{x_i - x\sub{fit}}{h}\right)K\left(\frac{x_i - x\sub{fit}}{h}\right)  \right\}
P_{k}\left(\frac{x - x\sub{fit}}{h}\right) \\
&=& \frac{1}{nh}\sum_{i=1}^{n}\underbrace{\left\{  \sum_{k=0}^{\infty} t(k)
  P_{k}\left(\frac{x - x\sub{fit}}{h}\right) P_{k}\left(\frac{x_i - x\sub{fit}}{h}\right) K\left(\frac{x_i - x\sub{fit}}{h}\right)  \right\}}_{K\sub{eff}\left(\frac{x-x_i}{h}\right)}\\
&=& \frac{1}{nh}\sum_{i=1}^{n}K\sub{eff}\left(\frac{x-x_i}{h}\right)
\end{eqnarray*}
Defining $y=(x\sub{fit}-x_i)/h$, evaluate $K\sub{eff}$ at grid point
$x\sub{fit}$ to see that
\[ K\sub{eff}\left(\frac{x\sub{fit}-x_i}{h}\right) \equiv
K\sub{eff}\left(y\right) = \sum_{k=0}^{\infty} t(k)
  P_{k}\left(0\right) P_{k}\left(-y\right) K\left(-y\right). \] 
To establish (i)--(iii), note that Assumptions (a) and (b) imply  that $P_{k}(x)$ is an even
(odd) function for any even (odd) integer $k$. This means $P_{k}(-x)=P_{k}(x)$ for $k$ even, and $P_{k}(0)=0$ for $k$
odd, so that the effective kernel becomes
\begin{equation}\label{eq:Keff-even}
 K\sub{eff}\left(x\right) = \sum_{\{k\,:\,k\geq 0,\ k\text{ even}\}} t(k)
  P_{k}\left(0\right) P_{k}\left(x\right) K\left(x\right),
\end{equation}
and $K\sub{eff}(-x)=K\sub{eff}(x)$ is an even function supported
also on $(-a_K,a_K)$, thus establishing (i). 
Now, multiplying both sides of the above
equation by $P_0(x) \equiv 1$ and integrating, gives
\begin{eqnarray*}
\int_{\R} K\sub{eff}(x)dx &=& \int_{-a_K}^{a_K} K\sub{eff}(x)P_0(x)dx \\
&=& \int_{\at}^{\bt} \sum_{\{k\,:\,k\geq 0,\ k\text{ even}\}} t(k)
P_{k}(0)P_0(x)P_{k}(x)K(x) dx, 
\end{eqnarray*}
which follows by Assumption (c). Now, interchanging integral and sum in
the above expression and then using (\ref{eq:norm}),  establishes (ii)
as follows: 
\begin{eqnarray*}
\int_{\R} K\sub{eff}(x)dx &=& \sum_{\{k\,:\,k\geq 0,\ k\text{ even}\}} t(k) P_{k}(0) \int_{\at}^{\bt}
P_0(x)P_{k}(x)K(x) dx \\
&=& \sum_{\{k\,:\,k\geq 0,\ k\text{ even}\}} t(k) P_{k}(0)\delta_{0\,k} \\
&=& t(0) P_{0}(0) = t(0).
\end{eqnarray*}
%
To prove (iii), first define the $j$-th kernel moment as 
\[ \mu_j(K\sub{eff}) \equiv \int_\R x^j K\sub{eff}(x)dx. \]
Now, since the effective kernel is an even
function, it is clear  $\mu_j(K\sub{eff})=0$ for $j$ odd.
Hence, it suffices to consider the case when $j$ is even, whence
\[ \mu_j(K\sub{eff}) = \int_\R x^j K\sub{eff}(x)dx =
\sum_{\{k\,:\,0\leq k\leq M,\ k\text{ even}\}}
P_{k}(0) \int_{-a_K}^{a_K} x^j P_{k}(x) K(x) dx = \sum_{\{k\,:\,0\leq k\leq M,\ k\text{ even}\}}\alpha_{jk} P_{k}(0), \]
if we define
\[ \alpha_{jk} = \int_{-a_K}^{a_K} x^j P_{k}(x)K(x) dx = \int_{\at}^{\bt} x^j P_{k}(x)K(x) dx.  \] Now, from the theory of orthogonal polynomials, we know that
\[ x^j = \sum_{k=0}^{j} a_{jk}P_{k}(x), \qquad\text{where}\qquad  a_{jk}= \int_{\at}^{\bt} x^j P_{k}(x)K(x) dx = \alpha_{jk}. \]
Since $\alpha_{jk}$ is the coefficient of the $P_{k}(x)$ contribution
(a polynomial of order $k$) to the series expansion of $x^j$, it is obvious that
$\alpha_{jk}=0$ for $k>j$, and $\alpha_{jk}=0$ when $k$ and $j$ have
opposite parity (only even $k$ terms contribute when $j$ is even, and
vice-versa). With these observations, it is clear that for $j\leq M$
\[ \mu_j(K\sub{eff}) = \sum_{k=0}^{M} \al_{jk}P_{k}(0),
\qquad\text{and}\qquad x^j = \sum_{k=0}^{M} \al_{jk}P_{k}(x), \]
whence we see that 
\begin{eqnarray*}
\mu_j(K\sub{eff}) &=&  \left.x^j\right|_{x=0} 
= \begin{cases}
   1, & j=0, \\
   0, & j=1,\ldots,M.
 \end{cases}
\end{eqnarray*}
If $M$ is even, then  since $M+1$
is odd and $K\sub{eff}(x)$ is an even
function, we have additionally that $\mu_{M+1}(K\sub{eff})=0$. Thus the effective kernel order is $M+1$ if $M$ is odd, and $M+2$ if $M$ is
      even.

\section{Proof of Theorem~\ref{th:osde}}
As $h \rightarrow \infty$ the value of the kernel $K(\cdot)$ becomes
less and less dependent on the grid point $\xf$ inside
$[a, b]$. In fact, starting from (\ref{eq:norm0}), note that for very large h, $K((x_i-\xf)/h)$ essentially becomes
constant on $[a, b]$. Equation (4) then gives rise to Legendre
polynomials since these are generated
when integrating with
respect to a constant weight function, in a manner similar to
Proposition~\ref{rk:gegen}. To see this, start with the orthonormal Legendre
polynomials $L_k(z)$ on  $[-1, 1]$, satisfying
\begin{equation}\label{legendre-ortho}  
\delta_{jk}=\int_{-1}^{1} L_{j}(z)L_{k}(z)dz . 
\end{equation}
To construct the corresponding orthonormal  system on $[a, b]$, we
make the transformation, $z=(2x-a-b)/(b-a)$, so that (\ref{legendre-ortho}) becomes
\begin{eqnarray}
\delta_{jk} &=& \int_{a}^{b}
\frac{2}{b-a}L_{j}\left(\frac{2x-a-b}{b-a}\right)L_{k}\left(\frac{2x-a-b}{b-a}\right)dx
= \int_{a}^{b} P_{j}\left(x\right)P_{k}\left(x\right)dx, \label{lege-delta}  
\end{eqnarray}
where 
\begin{equation}\label{temp-lege-poly}  
P_k(x)\equiv\sqrt{\frac{2}{b-a}}L_{k}\left(\frac{2x-a-b}{b-a}\right).
\end{equation}
Now construct an orthonormal  system on the interval $[\at,\bt]$ using $K(0)$
as the weight function instead of 1. By means of the transformation
$y=(x-\xf)/h$, (\ref{lege-delta}) then becomes 
\begin{eqnarray*}
\delta_{jk} &=& \int_{a}^{b}
\frac{1}{\sqrt{K(0)}}P_{j}\left(x\right)\frac{1}{\sqrt{K(0)}}P_{k}\left(x\right)K(0)
dx \\
&=& \int_{\at}^{\bt}
\sqrt{\frac{h}{K(0)}}P_{j}\left(yh+\xf\right)\sqrt{\frac{h}{K(0)}}P_{k}\left(yh+\xf\right)K(0)
dy \\
&=& \int_{\at}^{\bt} \tilde{P}_j(y)\tilde{P}_k(y)K(0)dy
\end{eqnarray*}
where
\begin{equation}  
\tilde{P}_k(y)\equiv\sqrt{\frac{2h}{(b-a)K(0)}}L_{k}\left(\frac{2yh+2\xf-a-b}{b-a}\right),
\end{equation}
which follows from (\ref{temp-lege-poly}). Now, from the proof of Theorem~\ref{th:lorpe-facts} we have the
following expression for LOrPE:
\[
\hflpe(x) = \frac{1}{nh}\sum_{i=1}^{n}\sum_{k=0}^{M}
  P_{k}\left(\frac{x - x\sub{fit}}{h}\right) P_{k}\left(\frac{x_i -
      x\sub{fit}}{h}\right) K\left(\frac{x_i - x\sub{fit}}{h}\right).
\]
Substituting $\tilde{P}_k(\cdot)$ for $P_k(\cdot)$ in the above
equation, gives
\begin{eqnarray}
\hflpe(x) &=& \frac{1}{nh}\sum_{i=1}^{n}\sum_{k=0}^{M}
  \tilde{P}_{k}\left(\frac{x - x\sub{fit}}{h}\right) \tilde{P}_{k}\left(\frac{x_i -
      x\sub{fit}}{h}\right) K\left(\frac{x_i - x\sub{fit}}{h}\right)
  \\
&=& \frac{1}{nh}\sum_{i=1}^{n}\sum_{k=0}^{M} \frac{2h}{(b-a)K(0)}
L_{k}\left(\frac{2\left(\frac{x-\xf}{h}\right)h+2\xf-a-b}{b-a}\right) \\
&& \qquad\qquad L_{k}\left(\frac{2\left(\frac{x_i-\xf}{h}\right)h+2\xf-a-b}{b-a}\right)K\left(\frac{x_i
    - x\sub{fit}}{h}\right) \\
&=& \frac{1}{n}\sum_{i=1}^{n}\sum_{k=0}^{M}\frac{2}{(b-a)K(0)}L_{k}\left(\frac{2x-a-b}{b-a}\right)L_{k}\left(\frac{2x_i-a-b}{b-a}\right)K\left(\frac{x_i
    - x\sub{fit}}{h}\right).
\end{eqnarray}
Since 
\[ \lim_{h \rightarrow \infty} K\left(\frac{x_i- x\sub{fit}}{h}\right) =
K(0), \]
we obtain
\begin{equation}
\hflpe(x) = \frac{1}{n}\sum_{i=1}^{n}\sum_{k=0}^{M}\sqrt{\frac{2}{b-a}}L_{k}\left(\frac{2x-a-b}{b-a}\right)\sqrt{\frac{2}{b-a}}L_{k}\left(\frac{2x_i-a-b}{b-a}\right), 
\end{equation}
which is the classical OSDE (\ref{classic-osde}) in terms of the orthogonal polynomials
\[ \phi_{k}(x)=\sqrt{\frac{2}{b-a}} L_k\left(\frac{2x-a-b}{b-a}\right). \] 


\begin{thebibliography}{aaaaa}

\bibitem{Bowman1984}
Bowman, A.W. (1984), ``An alternative method of cross-validation for the
 smoothing of density estimates'', \emph{Biometrika}, 71, 353--360.

\bibitem{Buja-etal-1989}
Buja, A., Hastie, T., and Tibshirani, R. (1989), 
``Linear Smoothers and Additive Models'', \emph{The Annals of
  Statistics} (with discussion),
17, 453--555.

\bibitem{SavitGolay1964}
\v{C}encov, N.N. (1962), ``Evaluation of an unknown distribution
density from observations'', \emph{Soviet Math. Dokl.}, 3, 1559--1562.

\bibitem{Charpentier2006}
Charpentier, A., Fermanian J.-D. and Scaillet, O. (2006),
``The Estimation of Copulas: Theory and Practice'' in:
Rank, J. (ed.), {\it Copulas: From Theory to Application in Finance},
35-60, Risk Books: London.

\bibitem{Chen1999}
Chen, S.X. (1999), ``Beta kernel estimators for density
functions'', \emph{Comp. Statist Data Anal.}, 31, 131--45.

\bibitem{ChenHuang2007}
Chen, S.X. and Huang, T.M. (2007), ``Nonparametric estimation of
copula functions for dependence modelling'',
{\it Canadian Journal of Statistics}, {\bf 35}, 265-282.

\bibitem{chiu92}
Chiu, S.-T. (1992), ``An Automatic Bandwidth Selector for Kernel
Density Estimation'', {\it Biometrika}, 79, 771-782.

\bibitem{CopasFryer80}
Copas,  J.B. and Fryer, M.J. (1980), ``Density Estimation and Suicide
Risks in Psychiatric Treatment'', \emph{Journal of the Royal
  Statistical Society}, Series A, 143, 167--176. 

\bibitem{DPADass2014}
Dassanayake, D.P.A. (2014), \emph{Local Orthogonal Polynomial Expansion and Empirical Sadddlepoint Approximation for Density Estimation}, PhD Dissertation, Texas Tech University, Lubbock.

\bibitem{dh86}
Diggle, P.J. and Hall, P. (1986), ``The selection of terms in an
orthogonal series density estimator'', {\it J. Americ. Statist. Assoc.}
{\bf 81}, 230-233.

\bibitem{duin1976}
Duin, R.P.W. (1976), ``On the choice of smoothing parameter for Parzen
estimators of probability density functions'' \emph{IEEE
Trans. Computers}, C-25, 1175--1179. 

\bibitem{Efrom99}
Efromovich, S. (1999), \emph{Nonparametric Curve Estimation: methods,
  theory, and applications}, New York: Springer. 

\bibitem{Elderton1969}
Elderton, W.P. and Johnson, N.L. (1969), \emph{Systems of Frequency Curves}, New
York: Cambridge University Press.

\bibitem{GijbelsMielniczuk1990}
Gijbels, I. and Mielniczuk, J. (1990), ``Estimating the Density of a Copula Function'', Communications in Statistics: Theory and Methods, {\bf 19}, 445-464.

\bibitem{givenshoeting2012}
Givens, G.H. and Hoeting, J.A. (2013), \emph{Computational
  Statistics}, 2nd ed., Hoboken: Wiley. 

\bibitem{Habbema-etal-1974}
Habbema, J.D.F., Hermans, J. and
  van der Broek, K. (1974), "A stepwise discrimination program using
  density estimation", in Bruckman, G. (ed.), \emph{Compstat 1974}, Vienna:
  Physica Verlag, 100--110.


\bibitem{hm88}
Hall, P. and Marron, J.S. (1988), ``Choice of kernel order in density estimation'', {\it The Annals of Statistics}, 16, 161-173.

\bibitem{HallTao02}
Hall, P. and Tao, T. (2002), ``Relative efficiencies of kernel and
local likelihood density estimators'', \emph{Journal of the Royal Statistical Society, Series B}, 64, 537-547.

\bibitem{Hall1983}
Hall, P. (1983), ``Large sample optimality of least squares
cross-validation in density estimation'', \emph{Ann. Statist.}, 11, 1156--1174.

\bibitem{Hart1985}
Hart, J.D. (1985), ``On the choice of a truncation point in Fourier
series density estimation'', \emph{J. Statist. Comput. Simulation}, 21, 95–-116. 

\bibitem{Heidenreich2013}
Heidenreich, N.-B., Schindler, A. and Sperlich, S. (2013),
``Bandwidth selection for kernel density
estimation: a review of fully automatic selectors'',
{\it Advances in Statistical Analysis} {\bf 97}, 403-433.

\bibitem{HjortJones96}
Hjort, N.L. and Jones, M.C. (1996), ``Locally parametric nonparametric density estimation'',
\emph{The Annals of Statistics}, 24, 1619-1647.

\bibitem{JOnesHender2007}
Jones, M.C. and Henderson, D.A. (2007), ``Kernel-type density estimation on the unit interval'', \emph{Biometrika}, 94, 977--984. 

\bibitem{Kakizawa2004}
Kakizawa, Y. (2004), ``Bernstein polynomial probability density
estimation'', \emph{J. Nonparametr. Stat.}, 16, 709–-729. 

\bibitem{llde96}
Loader, C.R. (1996), ``Local Likelihood Density Estimation'', {\it Ann. Statist.}
{\bf 24}, 1602-1618.

\bibitem{llde99}
Loader, C.R. (1999), \emph{Local Regression and Likelihood}, New
York: Springer. 

\bibitem{MalSch14}
Malec, P. and Schienle, M. (2014), ``Nonparametric kernel density
estimation near the boundary'', \emph{Computational Statistics and Data Analysis}, 72, 57--76. 

\bibitem{Schuster-Gregory-1981}
  Schuster, E.F. and Gregory, C.G. (1981), "On the inconsistency of maximum
  likelihood nonparametric density estimators", in Eddy, W.F. (ed.),
  \emph{Computer Science and Statistics: Proceedings of the 13th Symposium
  on the Interface}, New York: Springer-Verlag, 295--298. 


\bibitem{SiSe04}
Sheather, S.J. (2004), ``Density Estimation'', {\it Statist. Sci.}
{\bf 19}, 588-597.

\bibitem{SheJones91}
Sheather, S.J. and Jones, M.C. (1991), ``A reliable data-based
bandwidth selection method for kernel density estimation'', \emph{Journal of
the Royal Statistical Society}, series B, 53, 683-–690. 

\bibitem{Scott92}
Scott, D.W. (1992), \emph{Multivariate Density Estimation: Theory
  Practice and Visualization}, New York: Wiley. 

\bibitem{Silver1986}
Silverman, B.W. (1986), \textit{Density Estimation for Statistics and
  Data Analysis}, London: Chapman \& Hall. 

\bibitem{TharterLock1993}
Tarter, M.E. and Lock, M. (1993), \emph{Model-free Curve Estimation}, New York: Chapman \& Hall.

\bibitem{terrell90}
Terrell, G.R. (1990), ``The maximal smoothing principle in density estimation'', \emph{J. Amer. Statist. Assoc.}, 85, 470–-477. 

\bibitem{Thas2010}
Thas, O. (2010), \emph{Comparing Distributions}, New
York: Springer.

\bibitem{Volobouev2011}
Volobouev, I. (2011), ``Matrix Element Method in HEP: Transfer Functions,
Efficiencies, and Likelihood Normalization'',
arXiv:1101.2259 [physics.data-an].

\bibitem{Volobouev2012}
Volobouev, I. (2012), \emph{NPStat (Non-parametric Statistical Modeling
and Analysis)}; software available at http://npstat.hepforge.org http://npstat.hepforge.org.

\bibitem{WandJones95}
Wand, M. and Jones, M. (1995), \emph{Kernel Smoothing}, London: Chapman \& Hall. 

\bibitem{Larry06}
Wasserman, L. (2006), \emph{All of Nonparametric Statistics}, New
York: Springer. 

\bibitem{Wigmans2000}
Wigmans, R. (2000), \emph{Calorimetry: Energy Measurement in Particle
  Physics}, New
York: Oxford University Press.

\bibitem{YangMarron99}
Yang, L. and Marron, J.S. (1999), ``Iterated
Transformation-Kernel Density Estimation'', \emph{Journal of the American Statistical Association}, 94, 580--589. 

\bibitem{ZFan00}
Zhang, J. and Fan, J. (2000), ``Minimax kernels for nonparametric curve estimation'',
\emph{Journal of Nonparametric Statistics}, 12, 417-445.


\end{thebibliography}
\end{document}